\documentclass[12pt]{amsart}
\usepackage{amsmath}
\usepackage{amssymb}
\usepackage{latexsym}

\newcommand{\halmos}{\rule{5pt}{5pt}}

\newcommand{\cE}{{\mathcal E}}

\numberwithin{equation}{section}

\newtheorem{prop}{\bf Proposition}[section]

\theoremstyle{definition}

\begin{document}
\title[Degenerations of Ruijsenaars-van Diejen operator]
{Degenerations of Ruijsenaars-van Diejen operator and $q$-Painlev\'e equations}
\author{Kouichi Takemura}
\address{Department of Mathematics, Faculty of Science and Engineering, Chuo University, 1-13-27 Kasuga, Bunkyo-ku Tokyo 112-8551, Japan}
\email{takemura@math.chuo-u.ac.jp}
\begin{abstract}
It is known that the Painlev\'e VI is obtained by connection preserving deformation of some linear differential equations, and the Heun equation is obtained by a specialization of the linear differential equations.
We inverstigate degenerations of the Ruijsenaars-van Diejen difference opearators and show difference analogues of the Painlev\'e-Heun correspondence.
\end{abstract}
\maketitle

\section{Introduction}

In this paper, we investigate $q$-difference equations that are generalisations of the Heun equation and the Painlev\'e VI equation. 

Heun's differential equation is given by 
\begin{align}
\frac{d^2y}{dz^2} +  \left( \frac{\gamma}{z}+\frac{\delta }{z-1}+\frac{\epsilon}{z-t}\right) \frac{dy}{dz} + \frac{\alpha \beta z -q}{z(z - 1)(z - t)} y= 0,
\label{eq:Heun}
\end{align}
with the condition $\gamma +\delta +\epsilon = \alpha +\beta +1$, and it is a standard form of Fuchsian differential equation with four singularities $\{ 0,1,t,\infty \}$.
Note that the Gauss hypergeometric equation is a standard form of Fuchsian differential equation with three singularities $\{ 0,1,\infty \}$.
The Heun equation has an accessory parameter $q$ which is independent from local exponents, although the hypergeometric equation does not have it.

It is known that the Heun equation admits an expression in terms of elliptic functions.
Let $\wp (x)$ be the Weierstrass elliptic function with basic periods $(2\omega _1 ,2\omega _3)$.
Put $\omega _2=-\omega_1 -\omega _3$, $\omega _0=0$ and $e_i=\wp (\omega _i)$ $(i=1,2,3)$.
By setting 
\begin{equation}
z= \frac{\wp (x) -e_1}{e_2 - e_1},\; t=\frac{e_3 -e_1}{e_2-e_1}
\end{equation}
and applying a gauge transformation,
we obtain an elliptical representation of Heun's differential equation (see \cite{TakS}):
\begin{align}
& \left(-\frac{d^2}{dx^2} + \sum_{i=0}^3 l_i(l_i+1)\wp (x+\omega_i) \right)f(x)= Ef(x).
\label{eq:Heunellip}
\end{align}
Here the coupling constants $l_0, \dots ,l_3$ correspond to the parameters $\alpha , \dots ,\epsilon $ in Eq.(\ref{eq:Heun}) and the eigenvalue $E$ corresponds to the accessory parameter $q$.

The Painlev\'e VI equation is a non-linear ordinary differential equation given by
\begin{align}
 & \quad \frac{d^2 \lambda }{dt^2}=  \frac{1}{2} \left( \frac{1}{\lambda }+\frac{1}{\lambda -1}+\frac{1}{\lambda -t} \right) \left( \frac{d\lambda }{dt} \right) ^2 -\left( \frac {1}{t} +\frac {1}{t-1} +\frac {1}{\lambda -t} \right)\frac{d\lambda  }{dt} \\
& \qquad \qquad +\frac{\lambda  (\lambda -1)(\lambda -t)}{t^2(t-1)^2}\left\{ \alpha + \beta \frac{t}{\lambda  ^2} +\gamma \frac{(t-1)}{(\lambda -1)^2} +\delta \frac{t(t-1)}{(\lambda -t)^2} \right\}. \nonumber
\end{align}
See \cite{IKSY} for a review of the Painlev\'e equations.
In particular, it is known that solutions of the Painlev\'e VI equation do not have movable singularities other than poles, that is called the Painlev\'e property.
Painlev\'e VI is also obtained by monodromy preserving deformation of the $2 \times 2$ Fuchsian system of equations with four singularities $\{ 0,1,t,\infty \}$.
The Fuchsian system of equations is equivalent to the following Fuchsian equation 
\begin{align}
& \frac{d^2y_1}{dz^2} + \left( \frac{1-\theta _0}{z}+\frac{1-\theta _1}{z-1}+\frac{1-\theta _t}{z-t}-\frac{1}{z-\lambda}  \right)  \frac{dy_1}{dz} \label{eq:linP6} \\
& \qquad  + \left( \frac{\kappa _1(\kappa _2 +1)}{z(z-1)}+\frac{\lambda (\lambda -1)\mu}{z(z-1)(z-\lambda)}-\frac{t (t -1)H}{z(z-1)(z-t)}  \right) y_1=0, \nonumber \\
& H=\frac{1}{t(t-1)}[ \lambda (\lambda -1) (\lambda -t)\mu ^2 -\{ \theta _0  (\lambda -1) (\lambda -t) \nonumber \\
& \qquad +\theta _1  \lambda (\lambda -t) +(\theta _t -1) \lambda (\lambda -1)\} \mu +\kappa _1 (\kappa _2 +1) (\lambda -t)]. \nonumber
\end{align}
Note that the singularity $z=\lambda $ is apparent, which follows from the equality for $H$.
The monodromy of the solution to Eq.(\ref{eq:linP6}) is preserved as the parameter $t$ varies, if there exist rational functions $a_1(z,t)$ and $a_2 (z,t)$ of the variable $z$ such that the equation 
\begin{equation}
\frac{\partial y}{\partial t}= a_1(z,t) y+ a_2 (z,t) \frac{\partial y}{\partial z} \label{eq:dydt}
\end{equation}
is compatible to Eq.(\ref{eq:linP6}) (see \cite{IKSY}).
It follows from a lengthy calculation that the compatibility condition is equivalent to 
\begin{align}
\frac{d\lambda }{dt} =\frac{\partial H}{\partial \mu}, \quad \frac{d\mu }{dt} =-\frac{\partial H}{\partial \lambda},
\end{align}
which is called the Painlev\'e VI system.
By eliminating $\mu $, we obtain the Painlev\'e VI equation.

Recall that Eq.(\ref{eq:linP6}) has five singularities $\{ 0,1,t,\infty ,\lambda \} $, and the singularity $z=\lambda $ is superfluous for the Heun equation.
By specializing the point $z=\lambda $ to regular singularities $\{ 0,1,t,\infty \} $, we may derive the Heun equation.
For example, by setting $\lambda = t$ in Eq.(\ref{eq:linP6}) we have
\begin{align}
& \frac{d^2y_1}{dz^2} + \left( \frac{1-\theta _0}{z}+\frac{1-\theta _1}{z-1}+\frac{-\theta _t}{z-t} \right)  \frac{dy_1}{dz}+  \frac{\kappa _1 (\kappa _2 +1) (z-t) +\theta _t  t (t -1) \mu }{z(z-1)(z-t)} y_1=0 .
\end{align}
Therefore the Heun equation is related with the Painlev\'e VI equation through the linear differential equation given by Eq.(\ref{eq:linP6}).
We can also obtain the Heun equation by other specializations, and they are related with the space of initial conditions (see \cite{TakMH}).
See also \cite{SL,TakHP} for other perspectives on relationship between the Heun equation and the Painlev\'e VI equation.
Note that the Painlev\'e VI equation also admits elliptical representations \cite{Man,Guz,TakHP,ZZ}, which were applied in various ways.

In this paper, we propose a difference analogue of the correspondence between the Heun equation and the Painlev\'e VI equation.

Sakai \cite{Sak} investigated difference analogue of the Painlev\'e equation by using structures of some algebraic surfaces which are generalisations of the space of initial conditions, and proposed a list of the equations.
There are three kinds of difference Painlev\'e equations, i.e. elliptic difference, $q$-difference (or multiplicative difference) and additive difference, and each difference equation is labelled by some affine root systems from its symmetry.
The $q$-difference Painlev\'e equations of types $E^{(1)}_7$, $E^{(1)}_6$ and $D^{(1)}_5$ are at issue in this paper. 

Before giving a difference analogue of the Heun equation, we discuss a multivariable generalization of the Heun equation.
The quantum Inozemtsev system of type $BC_N$ is a quantum mechanical $N$-particle system whose Hamiltonian is given by
\begin{align}
& H=-\sum_{j=1}^N\frac{\partial ^2}{\partial x_j^2}+2l(l+1)\sum_{1\leq j<k\leq N} \left( \wp (x_j-x_k) +\wp (x_j +x_k) \right) \\
& \quad  \quad  \quad  + \sum_{j=1}^N \sum _{i=0}^3 l_i(l_i+1) \wp(x_j +\omega_i). \nonumber 
\end{align}
It is a generalization of the Calogero-Moser-Sutherland model, and the Inozemtsev model of type $BC_N$ is quantum Liouville integrable, i.e.~there exist operators
$H_k= \sum_{j=1}^N \left( \frac{\partial }{\partial x_j} \right) ^{2k} + \mbox{(lower terms)}$ such that $[H, H_k]=0$ and $[H_{k_1}, H_{k_2}]=0$ $(k, k_1, k_2=2,\dots ,N)$.
By restricting to the case $N=1$, we recover the elliptical representation of Heun's equation (see Eq.(\ref{eq:Heunellip})).

A difference (relativistic) analogue of the Inozemtsev system of type $BC_N$ is known as the Ruijsenaars-van Diejen system \cite{vD,RuiN} (or the Ruijsenaars system of type $BC_N$), whose defining second order difference operator is given by
\begin{align}
& A (\mu ; x ) = \sum _{j=1}^N ( V_{j}^+ (x) \exp(\delta \partial_{x_j}) + V_{j}^- (x) \exp(-\delta \partial_{x_j})  )+ V_0 (x) , \label{eq:RvDintro} \\
&  V_{j}^{\pm} (x) = \frac{ \prod_{s=1}^8 \theta (\pm x_j + \mu _s) }{\theta (\pm 2 x_j) \theta (\pm 2 x_j +\delta ) } \prod_{k\neq j}\frac{\theta (\pm x_j+x_k +\kappa )\theta (\pm x_j-x_k +\kappa )}{\theta (\pm x_j+x_k )\theta (\pm x_j-x_k )} ,\nonumber 
\end{align}
where $\theta (x)$ is the theta function and we omit the expression of the function $V_0(x)$, instead we give another explicit expression in Eq.(\ref{eq:defeRvDN}).
Note that 
\begin{align}
& \exp(\pm \delta \partial_{x_j}) f(x_1, \dots ,x_j, \dots ,x_N) = f(x_1, \dots ,x_j \pm \delta , \dots ,x_N).
\end{align}
The system contains the parameters $\delta , \kappa , \mu _1 ,\dots ,\mu _8$.
By a suitable limit as $\delta \to 0$, we obtain the Hamiltonian of the Inozemtsev system \cite{vD,RuiN}.
It is known that commuting operators of the Ruijsenaars-van Diejen system exist as is the case of the Inozemtsev system \cite{KH}.
We may regard the Ruijsenaars-van Diejen operator with one variable as a difference analogue of the Heun equation.
It is known that the Ruijsenaars-van Diejen operator has $E_8$ spectral symmetry \cite{Rui15}.
On the other hand, the elliptic difference Painlev\'e equation admits $E_8^{(1)}$ symmetry \cite{Sak}.
We expect to clarify relationships between the one variable difference equation of Ruijsenaars-van Diejen type and the elliptic difference Painlev\'e equation.

In this paper we investigate degenerations of the Ruijsenaars-van Diejen operator and find correspondences with linear $q$-difference equations which are related with $q$-difference Painlev\'e equations.
We find that we can take degenerations of the Ruijsenaars-van Diejen operator of $N$ variables four times, although it seems that the first two were essentially obtained by van Diejen \cite{vD}.
The degenerations are still interesting in the setting of one variable.
By taking degeneration four times, we obtain the following $q$-difference operator $ A^{\langle 4 \rangle} (x)$:
\begin{align}
& A^{\langle 4 \rangle} (x) g(x) =  x^{-1} (x-h_1 q^{1/2} ) (x-h_2 q^{1/2} )  g(x/q) + x^{-1}  l_3 l_4 (x - l_1 q^{-1/2} ) (x - l_2 q^{-1/2} ) g(qx) \\
& \qquad -\{ (l_3 +l_4 ) x  + (l_1 l_2 l_3 l_4 h_1 h_2 )^{1/2} ( h_3^{1/2} + h_3^{-1/2} ) x^{-1} \} g(x) . \nonumber
\end{align}
Then we may regard the equation 
\begin{equation}
A^{\langle 4 \rangle} (x) g(x) =Eg(x) \;  (E: \mbox{ eigenvalue})
\label{eq:A4E}
\end{equation}
as a $q$-deformation of Heun equation (\ref{eq:Heun}).
On the other hand, Eq.(\ref{eq:A4E}) is obtained as a special case of the linear $q$-difference equation by Jimbo and Sakai \cite{JS} which is related with the $q$-Painlev\'e VI equation by the connection preserving deformation. 
Similarly the equations for eigenfunctions of the second degenerate operator and the third degenerate operator are also obtained as special cases of the linear $q$-difference equations obtained by Yamada \cite{Y} which are related with the $q$-Painlev\'e equations of type $E_6^{(1)}$ and type $E_7^{(1)}$.

This article is organized as follows.
In section \ref{sec:degenone}, we apply degeneration of the Ruijsenaars-van Diejen operator with one variable four times.
In section \ref{sec:JSY}, we review linear $q$-difference equations which are related with some $q$-Painlev\'e equations and obtain the degenerated Ruijsenaars-van Diejen operators with one variable by specializing the parameters.
In section \ref{sec:degen}, we extend the degeneration to the multivariable case.
In section \ref{sec:discuss}, we propose some problems related with results in this paper.

\section{Degeneration of Ruijsenaars-van Diejen operator with one variable} \label{sec:degenone}

\subsection{Ruijsenaars-van Diejen operator}

We describe the Ruijsenaars-van Diejen operator with one variable explicitly.
Let $a_{+}$, $a_-$ be complex numbers whose real parts are positive and $R_{\pm } (z)$ be the functions defined by
\begin{equation}
R_{\pm }(z)=\prod_{k=1}^{\infty}(1-q_{\pm }^{2k-1} e^{2\pi i z}) (1-q_{\pm }^{2k-1} e^{-2\pi i z} ),\; q_{\pm} =e^{-\pi a_{\pm}}. \label{eq:R+-}
\end{equation}
They are modified versions of theta functions with the half periods $1/2 $ and $i a_{\pm}/2$.
The Ruijsenaars-van Diejen operator of one variable is given by 
\begin{equation}
A_{+}(h;z) = V_{+}(h;z)\exp(-ia_{-}\partial_{z})+V_{+}(h;-z)\exp(ia_{-}\partial_{z}) +U_{+}(h;z),
\label{eq:defeRvD}
\end{equation}
where
\begin{align}
& V_{+}(h;z) = \frac{\prod_{n=1}^8 R_{+}(z-h_n-ia_{-}/2)}{R_{+}(2z+ia_{+}/2)R_{+}(2z-ia_{-}+ia_{+}/2)}, \\
& U_{+}(h;z) = \frac{\sum_{t=0}^3p_{t,+}(h)
[ \cE_{t,+}(\mu ;z)-\cE_{t,+}(\mu ;\omega_{t,+})]}{2R_{+}(\mu -ia_{+}/2)R_{+}(\mu  -ia_{-}-ia_{+}/2)}, \nonumber
\end{align}
and we are using
\begin{align}
& \omega_{0,+}=0,\ \ \omega_{1,+}=1/2,\ \ \omega_{2,+}=ia_{+}/2,\ \ \omega_{3,+}=-1/2-ia_{+}/2, \\
& p_{0,+}(h) = \prod_{n=1}^8 R_{+}(h_n),\; p_{2,+}(h) = e^{-2\pi a_{+}}\prod_{n=1}^8 e^{-i\pi h_n } R_{+}(h_n-ia_{+}/2), \nonumber \\
& p_{1,+}(h) = \prod_{n=1}^8 R_{+}(h_n-1/2),\;  p_{3,+}(h) = e^{-2\pi a_{+}}\prod_{n=1}^8 e^{i\pi h_n } R_{+}(h_n+1/2+ia_{+}/2), \nonumber \\
& \cE_{t,+}(\mu ;z) = 
\frac{R_{+}(z+\mu  -ia_+ /2 -ia_- /2 -\omega_{t,+})R_{+}(z-\mu  +ia_+ /2 +ia_- /2 -\omega_{t,+})}{R_{+}(z-ia_+ /2 -ia_- /2 -\omega_{t,+})R_{+}(z+ia_+ /2 +ia_- /2 -\omega_{t,+})}, \nonumber
\end{align}
$(t=0,1,2,3).$
We adapt the expression in \cite{Rui15}, which is slightly different from the one in \cite{Rui09} with an additive constant.
Note that the function $U_{+}(h;z) $ is independent from the parameter $\mu $ in the case of one variable $z$, which can be proved as the first part of Lemma 3.2 in \cite{RuiN}.
Hence the operator $A_{+}(h;z)$ is also independent from the parameter $\mu $.

We can obtain an elliptical representation of the Heun equation (\ref{eq:Heunellip}) from the equation $A_{+}(h;z) f(z)= E f(z)$ ($E$: eigenvalue) by taking a suitable limit as $a_- \to 0$. For details see \cite{RuiN}.

\subsection{First degeneration}
We are going to take a trigonometric limit $(q_+  \to 0)$ of the Ruijsenaars-van Diejen operator with one variable.
The function $R_+(z)$ satisfies
\begin{equation}
R_+ (z \mp ia_+ ) =-e^{\pi a_+} e^{\pm 2\pi i z} R_+ (z) \label{eq:R+period}
\end{equation} 
and we have the following expansion as $q_+ \to 0$ (or $a_+ \to +\infty  $):
\begin{align}
& R_{+}(z )= 1-(e^{2\pi i z}+ e^{-2\pi i z} )q_{+ }+ q_{+ }^2 +O(q_+ ^3), \label{eq:R+expand} \\
& R_{+}(z \pm ia_+/2 )= (1- e^{\mp 2\pi i z} ) (1-(e^{2\pi i z}+ e^{-2\pi i z} )q_{+ }^{2} +O(q_+ ^4)). \nonumber
\end{align}

We set $h_n= \tilde{h}_n -ia_+/2$.
Then the function $V_{+}(h;z) $ admits the following limit as $q_+ \to 0$:
\begin{align}
& V_{+}(h;z) \to V^{\langle 1 \rangle}  (h;z) = \frac{\prod_{n=1}^8 (1-e^{-2\pi i z}e^{2\pi i \tilde{h}_n }e^{-\pi a_{-}})}{(1-e^{-4 \pi i z })(1-e^{-4\pi i z}e^{-2\pi a_{-}})}. \label{eq:tVhz}
\end{align}
By considering the limit of the function $U_{+}(h;z) $ as $q_+ \to 0$, we have the following proposition.
\begin{prop} \label{prop:firstdeg}
Let $ A(h,q_+ ;z) $ be the Ruijsenaars-van Diejen operator defined in Eq.(\ref{eq:defeRvD}).
As $q_+ \to 0$, we have
\begin{equation}
\left( A(h,q_+; z) + \frac{\prod_{n=1}^8 e^{\pi i \tilde{h}_n } }{(1- e^{\pi a_- } )^2}q_+^{-2} +C \right) f(z) \to A^{\langle 1 \rangle}  (h;z) f(z)
\end{equation}
for any $f(z)$, where
\begin{equation}
A^{\langle 1 \rangle}  (h;z) = V^{\langle 1 \rangle}  (h;z)\exp(-ia_{-}\partial_{z})+ V^{\langle 1 \rangle} (h;-z)\exp(ia_{-}\partial_{z}) )+U^{\langle 1 \rangle} (h;z),
\end{equation}
$V^{\langle 1 \rangle} (h;z)$ was defined in Eq.(\ref{eq:tVhz}),
\begin{align}
& U^{\langle 1 \rangle} (h;z) = \frac{\prod_{n=1}^8 (e^{2 \pi i \tilde{h}_n }- 1)}{2(1-e^{2\pi i z}e^{\pi a_- })(1-e^{-2\pi i z}e^{\pi a_- })} + \frac{\prod_{n=1}^8 (e^{2 \pi i \tilde{h}_n } + 1) }{2(1+e^{2\pi i z}e^{\pi a_- })(1+e^{-2\pi i z}e^{\pi a_- })} \label{eq:Vbdeg1} \\
&  + e^{-\pi a_- } \prod_{n=1}^8 e^{\pi i \tilde{h}_n } \cdot \Big[ (e^{2\pi i z} +  e^{-2\pi i z} ) \sum _{n=1}^8  ( e^{2\pi i \tilde{h}_n }+ e^{-2\pi i \tilde{h}_n }) - (e^{\pi a_- }+e^{-\pi a_- } ) (e^{4\pi i z} + e^{-4\pi i z} ) \Big] , \nonumber
\end{align}
and
\begin{align}
& C= \prod_{n=1}^8 e^{\pi i \tilde{h}_n } \cdot \Big[ e^{-\pi a_- } (e^{\pi a_- }+e^{-\pi a_- }) \\
&  + \frac{\displaystyle 12 +  \sum _{1\leq n <n' \leq 8} (e^{2\pi i \tilde{h}_n } + e^{-2\pi i \tilde{h}_n })( e^{2\pi i \tilde{h}_{n'} } + e^{-2\pi i \tilde{h}_{n'} }  )  + \frac{1}{2} \Big\{ \prod_{n=1}^8 (e^{2\pi i \tilde{h}_n }- 1) +\prod_{n=1}^8 (e^{2\pi i \tilde{h}_n } +1 ) \Big\} }{(1- e^{\pi a_- } )^2}\Big] . \nonumber
\end{align} 
\end{prop}
\begin{proof}
It follows from $R_{+}(z \pm ia_+/2 )= 1- e^{\mp 2\pi i z} +O(q_{+ }^{2})$ that  
\begin{align}
& p_{0,+}(h) \cE_{0,+}(\mu ;z) \\
& =  \frac{R_{+}(z+\mu  -ia_+ /2 -ia_- /2)R_{+}(z-\mu  +ia_+ /2 +ia_- /2 )}{R_{+}(z-ia_+ /2 -ia_- /2 )R_{+}(z+ia_+ /2 +ia_- /2)} \prod_{n=1}^8 R_{+}(\tilde{h}_n -ia_+/2) \nonumber \\
& =  \frac{(1-e^{2\pi i (z+\mu  -ia_- /2)})(1-e^{-2\pi i (z-\mu  +ia_- /2 )})}{(1-e^{2\pi i (z -ia_- /2)})(1-e^{-2 \pi i (z +ia_- /2 )})} \prod_{n=1}^8 (1- e^{2 \pi i \tilde{h}_n })+ O(q_+^2) \nonumber \\
& = \Big\{ e^{2\pi i \mu  } +\frac{(e^{2\pi i \mu  } -1)(e^{2\pi i \mu }e^{2 \pi a_- } -1)}{(1-e^{2\pi i z}e^{\pi a_- })(1-e^{-2\pi i z}e^{\pi a_- })} \Big\}  \prod_{n=1}^8  (e^{2 \pi i \tilde{h}_n }-1) + O(q_+^2) , \nonumber 
\end{align}
\begin{align}
& p_{1,+}(h) \cE_{1,+}(\mu ;z)  \\
& = \Big\{ e^{2\pi i \mu  } +\frac{(e^{2\pi i \mu  } -1)(e^{2\pi i \mu }e^{2 \pi a_- } -1)}{(1+e^{2\pi i z}e^{\pi a_- })(1+e^{-2\pi i z}e^{\pi a_- })} \Big\} \prod_{n=1}^8 (e^{2\pi i \tilde{h}_n }+ 1)  + O(q_+^2) ,\nonumber \\
& R_{+}(\mu -ia_{+}/2) R_{+}(\mu  -ia_{-}-ia_{+}/2)  = (1- e^{2\pi i \mu } )(1- e^{2\pi i \mu }e^{2\pi a_{-}} ) +O(q_+ ^2). \nonumber 
\end{align}
Hence
\begin{align} 
& \frac{p_{0,+}(h) (\cE_{0,+}(\mu ;z)- \cE_{0,+}(\mu ;\omega_{0,+}))}{2R_{+}(\mu  -ia_{+}/2)R_{+}(\mu  -ia_{-}-ia_{+}/2) } + \frac{p_{1,+}(h) (\cE_{1,+}(\mu ;z)- \cE_{1,+}(\mu ;\omega_{1,+}))}{2R_{+}(\mu  -ia_{+}/2)R_{+}(\mu  -ia_{-}-ia_{+}/2) } \\
& =  \Big\{ \frac{1}{2(1-e^{2\pi i z}e^{\pi a_- })(1-e^{-2\pi i z}e^{\pi a_- })} - \frac{1}{2(1-e^{\pi a_- })^2} \Big\} \prod_{n=1}^8 (e^{2 \pi i \tilde{h}_n }-1) \nonumber  \\
& +  \Big\{ \frac{1}{2(1+e^{2\pi i z}e^{\pi a_- })(1+e^{-2\pi i z}e^{\pi a_- })} - \frac{1}{2(1-e^{\pi a_- })^2} \Big\} \prod_{n=1}^8 (e^{2 \pi i \tilde{h}_n }+1) +O(q_+ ^2) . \nonumber 
\end{align}
Set
\begin{align}
& \tilde{p}(h)= 8+ \sum _{1\leq n <n' \leq 8} (e^{2\pi i \tilde{h}_n } + e^{-2\pi i \tilde{h}_n })( e^{2\pi i \tilde{h}_{n'} } + e^{-2\pi i \tilde{h}_{n'} }  ) .
\end{align}
It follows from Eqs.(\ref{eq:R+period}, \ref{eq:R+expand}) that 
\begin{align}
& p_{2,+}(h) = e^{-2\pi a_{+}} \prod_{n=1}^8 e^{-i\pi( \tilde{h}_n -ia_+/2)} R_{+}(\tilde{h}_n -ia_+) = e^{2\pi a_{+}} \prod_{n=1}^8 e^{ \pi i  \tilde{h}_n }R_{+}(\tilde{h}_n ) \\
& = e^{2\pi a_{+}} \prod_{n=1}^8 e^{ \pi i  \tilde{h}_n } \{ (1- q_+ (e^{2\pi i \tilde{h}_n } + e^{-2\pi i \tilde{h}_n } ) + q_+^2 +O(q_+^3) \} \nonumber  \\
& = q_+^{-2} \Big[ 1 - \sum _{n=1}^8 ( e^{2\pi i \tilde{h}_n }+ e^{-2\pi i \tilde{h}_n })q_+ + \tilde{p}(h) q_+^2 +O(q_+^3) \Big] \prod_{n=1}^8 e^{ \pi i  \tilde{h}_n }, \nonumber \\
& p_{3,+}(h)  = e^{2\pi a_{+}} \prod_{n=1}^8 e^{ \pi i  \tilde{h}_n }R_{+}(\tilde{h}_n +1/2) \nonumber \\
& = q_+^{-2} \Big[ 1 + \sum _{n=1}^8 ( e^{2\pi i \tilde{h}_n }+ e^{-2\pi i \tilde{h}_n }) q_+ +\tilde{p}(h) q_+^2  +O(q_+^3) \Big] \prod_{n=1}^8 e^{ \pi i  \tilde{h}_n } , \nonumber 
\end{align}
and we have
\begin{align}
& \cE_{2,+}(\mu ;z)  
= \frac{R_{+}(z+\mu -ia_- /2 -ia_{+})R_{+}(z-\mu  +ia_- /2)}{R_{+}(z-ia_- /2 -ia_{+})R_{+}(z+ia_- /2)} \\
& = e^{2\pi i \mu } \frac{R_{+}(z+\mu -ia_- /2 )R_{+}(z-\mu  +ia_- /2)}{R_{+}(z-ia_- /2 )R_{+}(z+ia_- /2)} \nonumber \\
& = e^{2\pi i \mu } \frac{(1-(e^{2\pi i (z+\mu -ia_- /2 )}+ e^{-2\pi i (z+\mu -ia_- /2)} )q_{+ } +q_+ ^2 +O(q_+ ^3))}{(1-(e^{2\pi i (z -ia_- /2 )}+ e^{-2\pi i (z -ia_- /2)} )q_{+ } +q_+ ^2 +O(q_+ ^3))} \cdot \nonumber \\
& \qquad \qquad \cdot \frac{(1-(e^{2\pi i (z-\mu +ia_- /2 )}+ e^{-2\pi i (z-\mu +ia_- /2)} )q_{+ } +q_+ ^2 +O(q_+ ^3))}{(1-(e^{2\pi i (z +ia_- /2 )}+ e^{-2\pi i (z+ia_- /2)} )q_{+ } +q_+ ^2 +O(q_+ ^3))} \nonumber \\
& = e^{2\pi i \mu  } [  1 - (e^{\pi i \mu  }  -e^{-\pi i \mu  }) (e^{\pi i \mu  } e^{\pi a_- } - e^{-\pi i \mu  }e^{-\pi a_- }) (e^{2\pi i z} +  e^{-2\pi i z} ) q_+  - \cE (\mu ;z) q_+ ^2+O(q_+ ^3)] , \nonumber \\
& \cE_{3,+}(\mu ;z) =\frac{R_{+}(z+\mu -ia_- /2 +1/2  )R_{+}(z-\mu  +ia_- /2 +1/2 + ia_{+} )}{R_{+}(z -ia_- /2 +1/2  )R_{+}(z+ia_- /2 +1/2 + ia_+ )} = \cE_{2,+}(\mu ;z + 1/2) \nonumber \\
& = e^{2\pi i \mu  } [  1 + (e^{\pi i \mu  }  -e^{-\pi i \mu  }) (e^{\pi i \mu  } e^{\pi a_- } - e^{-\pi i \mu  }e^{-\pi a_- }) (e^{2\pi i z} +  e^{-2\pi i z} ) q_+  - \cE (\mu ;z) q_+ ^2+O(q_+ ^3)] , \nonumber 
\end{align}
where
\begin{align}
& \cE (\mu ;z) =  (e^{\pi i \mu  }  -e^{-\pi i \mu  }) (e^{\pi i \mu  } e^{\pi a_- } - e^{-\pi i \mu  }e^{-\pi a_- })\{ (e^{4\pi i z} +  e^{-4\pi i z} )(e^{-\pi a_- }+e^{\pi a_- } ) \\
& - (e^{\pi i \mu  }  -e^{-\pi i \mu  }) (e^{\pi i \mu  } e^{\pi a_- } - e^{-\pi i \mu  }e^{-\pi a_- })\} . \nonumber 
\end{align}
Therefore
\begin{align}
& p_{2,+}(h) \cE_{2,+}(\mu ;z) + p_{3,+}(h) \cE_{3,+}(\mu ;z)  = 2 e^{2\pi i \mu  } \prod_{n=1}^8 e^{ \pi i  \tilde{h}_n } \cdot \Big[ q_+^{-2} + \tilde{p}(h) - \cE (\mu ;z) \\
& + (e^{\pi i \mu  }  -e^{-\pi i \mu  }) (e^{\pi i \mu  } e^{\pi a_- } - e^{-\pi i \mu  }e^{-\pi a_- }) (e^{2\pi i z} +  e^{-2\pi i z} ) \sum _{n=1}^8 ( e^{2\pi i \tilde{h}_n }+ e^{-2\pi i \tilde{h}_n })  \Big] + O(q_+) . \nonumber 
\end{align}
For $t=0,1,2,3$, we have
\begin{align}
& \cE_{t,+}(\mu ;\omega_{t,+}) = \frac{R_{+}(-\mu  +ia_- /2 + ia_{+} /2)^2}{R_{+}(ia_- /2 + ia_{+}/2 )^2}\\
& = \frac{ (1- e^{2\pi i \mu }e^{\pi a_- } )^2}{(1- e^{\pi a_- } )^2}  [1+2 (e^{\pi a_- } - e^{-\pi a_- } e^{-2\pi i \mu } )(1 - e^{2\pi i \mu })  q_{+ }^{2} +O(q_+ ^4)] .  \nonumber 
\end{align}
Then
\begin{align}
& p_{2,+}(h) \cE_{2,+}(\mu ;\omega_{2,+}) + p_{3,+}(h) \cE_{3,+}(\mu ;\omega_{3,+}) \\ 
& = 2 \frac{ (1- e^{2\pi i \mu }e^{\pi a_- } )^2}{(1- e^{\pi a_- } )^2}  \prod_{n=1}^8 e^{ \pi i  \tilde{h}_n } \cdot [ q_+^{-2}  +\tilde{p}(h) + 2 (e^{\pi a_- } - e^{-\pi a_- } e^{-2\pi i \mu } )(1 - e^{2\pi i \mu })   +O(q_+)] . \nonumber 
\end{align}
By combining with 
\begin{align}
& R_{+}(\mu -ia_{+}/2) R_{+}(\mu  -ia_{-}-ia_{+}/2) \\
& = (1- e^{2\pi i \mu } )(1- e^{2\pi i \mu }e^{2\pi a_{-}} ) (1-(e^{2\pi a_{-}} e^{2\pi i \mu }+ e^{-2\pi a_{-}} e^{-2\pi i \mu } +e^{2\pi i \mu }+ e^{-2\pi i \mu } )q_{+ }^{2} +O(q_+ ^4)) ,  \nonumber 
\end{align}
we have
\begin{align}
& \frac{p_{2,+}(h) (\cE_{2,+}(\mu ;z)- \cE_{2,+}(\mu ;\omega_{2,+})) + p_{3,+}(h) ( \cE_{3,+}(\mu ;z) -\cE_{3,+}(\mu ;\omega_{3,+}))}{2R_{+}(\mu -ia_{+}/2) R_{+}(\mu  -ia_{-}-ia_{+}/2) } \\ 
& = \prod_{n=1}^8 e^{ \pi i  \tilde{h}_n } \cdot \Big[ e^{-\pi a_- } \Big\{ (e^{2\pi i z} +  e^{-2\pi i z} )\sum _{n=1}^8 ( e^{2\pi i \tilde{h}_n }+ e^{-2\pi i \tilde{h}_n }) 
 - (e^{-\pi a_- }+e^{\pi a_- } ) (e^{4\pi i z} +  e^{-4\pi i z} ) \Big\}  \nonumber \\
&  - \frac{q_+^{-2}}{(1- e^{\pi a_- } )^2} - \frac{ \tilde{p}(h) +4}{(1- e^{\pi a_- } )^2} - e^{-\pi a_- } (e^{-\pi a_- }+e^{\pi a_- }) + O(q_+) \Big] .  \nonumber 
\end{align} 
Therefore
\begin{align}
& \frac{\sum _{k=0}^3 p_{k,+}(h) (\cE_{k,+}(\mu ;z)- \cE_{k,+}(\mu ;\omega_{k,+})) }{2R_{+}(\mu -ia_{+}/2) R_{+}(\mu  -ia_{-}-ia_{+}/2) } \\ 
& = \frac{\prod_{n=1}^8 e^{i\pi \tilde{h}_n} (e^{\pi i \tilde{h}_n }- e^{-\pi i \tilde{h}_n }) }{2(1-e^{2\pi i z}e^{\pi a_- })(1-e^{-2\pi i z}e^{\pi a_- })}+ \frac{\prod_{n=1}^8 e^{i\pi \tilde{h}_n} (e^{\pi i \tilde{h}_n }+ e^{-\pi i \tilde{h}_n }) }{2(1+e^{2\pi i z}e^{\pi a_- })(1+e^{-2\pi i z}e^{\pi a_- })}  \nonumber  \\
& + \prod_{n=1}^8 e^{ \pi i  \tilde{h}_n } \Big[ e^{-\pi a_- } (e^{2\pi i z} +  e^{-2\pi i z} )  \sum _{n=1}^8 ( e^{2\pi i \tilde{h}_n }+ e^{-2\pi i \tilde{h}_n })  -  e^{-\pi a_- } (e^{-\pi a_- }+e^{\pi a_- } )(e^{4\pi i z} +  e^{-4\pi i z} ) \nonumber \\
&  - \frac{q_+^{-2}}{(1- e^{\pi a_- } )^2}  - e^{-\pi a_- } (e^{-\pi a_- }+e^{\pi a_- })  \nonumber \\
&  - \frac{\displaystyle 12 +  \sum _{1\leq n <n' \leq 8} (e^{2\pi i \tilde{h}_n } + e^{-2\pi i \tilde{h}_n })( e^{2\pi i \tilde{h}_{n'} } + e^{-2\pi i \tilde{h}_{n'} }  )  + \frac{1}{2} \Big\{ \prod_{n=1}^8 (e^{2\pi i \tilde{h}_n }- 1) +\prod_{n=1}^8 (e^{2\pi i \tilde{h}_n } +1 ) \Big\} }{(1- e^{\pi a_- } )^2} \Big] .  \nonumber 
\end{align} 
\end{proof}
Note that the operator $A^{\langle 1 \rangle}  (h;z) $ does not contain the parameter $\mu $, and it is not surprising because the Ruijsenaars-van Diejen operator $ A(h,q_+ ;z) $ is independent from the parameter $\mu $ despite it appears in the expression.

To obtain the second degeneration in section \ref{subsec:secdeg}, we apply a gauge transformation to the operator in Proposition \ref{prop:firstdeg} by using the function $R_-(z)$ defined in Eq.(\ref{eq:R+-}), which satisfies
\begin{equation}
R_- (z \mp ia_- ) =-e^{\pi a_-} e^{\pm 2\pi i z} R_- (z).
\end{equation} 
By the gauge transformation 
\begin{equation}
\tilde{A}^{\langle 1 \rangle} (h,z)  = R_- (z)^{-2} \circ A^{\langle 1 \rangle} (h,z) \circ  R_- (z)^{2} ,
\end{equation}
we have the following operator:
\begin{equation}
\tilde{A}^{\langle 1 \rangle} (h;z) = \tilde{V}^{\langle 1 \rangle} (h;z)\exp(-ia_{-}\partial_{z})+\tilde{W}^{\langle 1 \rangle}(h;z)\exp(ia_{-}\partial_{z}) + U^{\langle 1 \rangle} (h;z),
\label{eq:N1E8}
\end{equation}
where $U^{\langle 1 \rangle} (h;z)$ was defined in Eq.(\ref{eq:Vbdeg1}) and
\begin{align}
& \tilde{V}^{\langle 1 \rangle} (h;z) = \frac{\prod_{n=1}^8 (1-e^{-2\pi i z}e^{2\pi i \tilde{h}_n }e^{-\pi a_{-}})}{e^{-2\pi a_{-}} e^{-4 \pi i z }(1-e^{-4 \pi i z })(1-e^{-4\pi i z}e^{-2\pi a_{-}})} , \\
& \tilde{W}^{\langle 1 \rangle} (h;z) = \frac{\prod_{n=1}^8 (1-e^{2\pi i z}e^{2\pi i \tilde{h}_n }e^{-\pi a_{-}})}{e^{-2\pi a_{-}} e^{4 \pi i z}(1-e^{4 \pi i z})(1-e^{4\pi i z}e^{-2\pi a_{-}})} . \nonumber 
\end{align}
Note that this operator was essentially obtained by van Diejen \cite{vD} in the multivariable case.
\subsection{Second degeneration} \label{subsec:secdeg}

We investigate a degeneration of the operator given by Eq.(\ref{eq:N1E8}).
\begin{prop} \label{prop:secdeg}
In Eq.(\ref{eq:N1E8}), we replace $z$ by $z+iR$, $\tilde{h}_n $ $(n=1,2,3,4)$ by $h_n +iR$, $\tilde{h}_n $ $(n=5,6,7,8)$ by $h_n -iR$ and take the limit $R\to +\infty $.
Then we arrive at the operator
\begin{equation}
A^{\langle 2 \rangle} (h;z) =  V^{\langle 2 \rangle} (h;z)\exp(-ia_{-}\partial_{z})+W^{\langle 2 \rangle} (h;z)\exp(ia_{-}\partial_{z}) +U^{\langle 2 \rangle} (h;z),
\end{equation}
where
\begin{align}
& V^{\langle 2 \rangle} (h;z) = e^{4\pi i  z} \prod_{n=1}^4 (1-e^{-2\pi i  z}e^{2\pi i h_n } e^{-\pi a_{-}})\prod_{n=5}^8 e^{2\pi i h_n} , \\
& W^{\langle 2 \rangle} (h;z) = e^{2\pi a_{-}} e^{-4\pi i z} \prod_{n=5}^8 (1-e^{2\pi i z}e^{2\pi i h_n} e^{-\pi a_{-}}) , \nonumber \\
& U^{\langle 2 \rangle} (h;z) = \prod_{n=5}^8 e^{2\pi i h_n } \Big[ \Big(\sum _{n=1}^4 e^{2\pi i h_n }+\sum _{n=5}^8 e^{-2\pi i h_n } \Big) e^{-\pi a_- } e^{2\pi i z} -(1+e^{-2\pi a_- })e^{4\pi i z} \Big] \nonumber \\
& \qquad + \prod_{n=1}^8 e^{\pi i h_n } \Big[ \Big(\sum _{n=1}^4  e^{-2\pi i h_n } + \sum _{n=5}^8  e^{2\pi i h_n } \Big) e^{-\pi a_- }  e^{-2\pi i z} - (1+e^{-2\pi a_- } ) e^{-4\pi i z} \Big] .\nonumber 
\end{align}
Namely we have
\begin{equation}
e^{-4\pi R} \tilde{A}^{\langle 1 \rangle} (h +iR v ;z +iR) f(z) \to A^{\langle 2 \rangle} (h;z) f(z)
\end{equation}
as $R\to +\infty $ for any $f(z)$, where $v = (1,1,1,1,-1,-1,-1,-1)$.
\end{prop}
\begin{proof}
We define the equivalence $a\sim b$ by $\lim _{R \to +\infty} a/b =1$.
Then
\begin{align}
& \tilde{V}^{\langle 1 \rangle} (h +iRv ;z +iR ) \\
& = \frac{\prod_{n=1}^4 (1-e^{-2\pi i z}e^{2\pi i h_n } e^{-\pi a_{-}})\prod_{n=5}^8 (1-e^{-2\pi i z}e^{4\pi R }e^{2\pi i h_n} e^{-\pi a_{-}})}{e^{-2\pi a_{-}} e^{-4\pi i z}e^{4\pi R}(1-e^{-4\pi i z}e^{4\pi R})(1-e^{-4\pi i z}e^{4\pi R}e^{-2\pi a_{-}})} \nonumber \\
& \sim  e^{4\pi i z}e^{4\pi R } \prod_{n=1}^4 (1-e^{-2\pi i z}e^{2\pi i h_n } e^{-\pi a_{-}})\prod_{n=5}^8 e^{2\pi i h_n} ,\nonumber 
\end{align}
\begin{align}
& \tilde{W}^{\langle 1 \rangle} (h +iRv ;z +iR ) \\
& = \frac{\prod_{n=1}^4 (1-e^{2\pi i z}e^{2\pi i h_n }e^{-4\pi R} e^{-\pi a_{-}})\prod_{n=5}^8 (1-e^{2\pi i z}e^{2\pi i h_n} e^{-\pi a_{-}})}{e^{-2\pi a_{-}} e^{4\pi i z}e^{-4\pi R}(1-e^{4\pi i z}e^{-4\pi R})(1-e^{4\pi i z}e^{-4\pi R}e^{-2\pi a_{-}})} \nonumber \\
& \sim e^{2\pi a_{-}} e^{-4\pi i z}e^{4\pi R} \prod_{n=5}^8 (1-e^{2\pi i z}e^{2\pi i h_n} e^{-\pi a_{-}}) , \nonumber 
\end{align}
\begin{align}
& \tilde{U}^{\langle 1 \rangle} (h +iRv ;z +iR ) \\
& = -e^{6\pi R} e^{-\pi a_- }e^{2\pi i z}  \prod_{n=5}^8 e^{2\pi i h_n } \frac{\prod_{n=1}^4 (1- e^{2\pi i h_n }e^{-2\pi R})\prod_{n=5}^8 (1- e^{-2\pi R} e^{-2\pi i h_n })}{2(1-e^{2\pi i z}e^{-2\pi R }e^{\pi a_- })(1-e^{-\pi a_- }e^{2\pi i z}e^{-2\pi R } )} \nonumber \\
& +e^{6\pi R} e^{-\pi a_- }e^{2\pi i z}  \prod_{n=5}^8 e^{2\pi i h_n } \frac{\prod_{n=1}^4 (1+ e^{2\pi i h_n }e^{-2\pi R})\prod_{n=5}^8 (1+ e^{-2\pi R} e^{-2\pi i h_n })}{2(1+e^{2\pi i z}e^{-2\pi R }e^{\pi a_- })(1+e^{-\pi a_- }e^{2\pi i z}e^{-2\pi R } )}  \nonumber \\
& + e^{-\pi a_- } \prod _{n=1}^8 e^{\pi i h_n } \Big[ \Big\{ \sum _{n=1}^4 ( e^{2\pi i h_n }e^{-2\pi R}+ e^{-2\pi i h_n } e^{2\pi R}) + \sum _{n=5}^8 ( e^{2\pi i h_n }e^{2\pi R}+ e^{-2\pi i h_n }e^{-2\pi R}) \Big\} \nonumber \\
&  \cdot (e^{2\pi i z}e^{-2\pi R} + e^{-2\pi i z}e^{2\pi R} ) - (e^{-\pi a_- }+e^{\pi a_- } )(e^{4\pi i z}e^{-4\pi R} + e^{-4\pi i z}e^{4\pi R} ) \Big] \nonumber \\
& \sim 0 \cdot e^{6\pi R} +  e^{4\pi R} \prod_{n=5}^8 e^{2\pi i h_n } \Big\{ \Big( \sum _{n=1}^4 e^{2\pi i h_n }+\sum _{n=5}^8 e^{-2\pi i h_n }\Big) e^{-\pi a_- }e^{2\pi i z} -(1+e^{-2\pi a_- })e^{4\pi i z} \Big\} \nonumber \\
& + e^{4\pi R} e^{-\pi a_- } \prod _{n=1}^8 e^{\pi i h_n } \Big\{ \Big( \sum _{n=1}^4  e^{-2\pi i h_n } + \sum _{n=5}^8  e^{2\pi i h_n } \Big) e^{-2\pi i z} - (e^{-\pi a_- }+e^{\pi a_- } ) e^{-4\pi i z} \Big\} .\nonumber 
\end{align}
Thus the proposition is obtained.
\end{proof}
Set $l_{n}= -h_{n+4}$ $(n=1,2,3,4) $.
By the multiplication and the gauge transformation given by
\begin{equation}
\tilde{A}^{\langle 2 \rangle} (h,l;z)  = e^{\pi a_- } \prod_{n=5}^8 e^{-2\pi i h_n}  \cdot e^{\pi i z} \circ A^{\langle 2 \rangle} (h,z) \circ e^{-\pi i z}, 
\end{equation}
we have the following operator:
\begin{align}
\tilde{A}^{\langle 2 \rangle} (h,l;z) =  \tilde{V}^{\langle 2 \rangle} (h;z)\exp(-ia_{-}\partial_{z})+\tilde{W}^{\langle 2 \rangle} (l;z)\exp(ia_{-}\partial_{z}) + \tilde{U}^{\langle 2 \rangle} (h,l;z),
\label{eq:N1E7}
\end{align}
where
\begin{align}
& \tilde{V}^{\langle 2 \rangle} (h;z) = e^{4\pi i  z} \prod_{n=1}^4 (1-e^{2\pi i h_n } e^{-\pi a_{-}}e^{-2\pi i  z}) , \;  \tilde{W}^{\langle 2 \rangle} (l;z) = e^{4\pi i z} \prod_{n=1}^4 (1- e^{2\pi i l_n} e^{\pi a_{-}} e^{-2\pi i z} ),\\
& \tilde{U}^{\langle 2 \rangle} (h,l;z) = \sum _{n=1}^4 ( e^{2\pi i h_n }+ e^{2\pi i l_n })  e^{2\pi i z} -(e^{\pi a_- }+e^{-\pi a_- })  e^{4\pi i z} \nonumber \\
& \qquad  + \prod_{n=1}^4 e^{\pi i (h_n + l_n) } \Big[ \sum _{n=1}^4 ( e^{-2\pi i h_n } + e^{-2\pi i l_n } ) e^{-2\pi i z}  - (e^{\pi a_- }+e^{-\pi a_- } ) e^{-4\pi i z} \Big] . \nonumber 
\end{align}
Note that this operator was also essentially obtained in \cite{vD}.

Set $x= e^{2\pi i z}$, $q= e^{-2\pi a_{-}}$, and replace $e^{2\pi i h_n }$ and $e^{2\pi i l_n }$ by $h_n$ and $l_n$.
Then the difference operator $\tilde{A}^{\langle 2 \rangle} (h,l;z) $ is written as
\begin{align}
& A^{\langle 2 \rangle} (x) g(x) = x^{-2} \prod_{n=1}^4 (x-h_n  q^{1/2} )  g(x/q) + x^{-2} \prod_{n=1}^4 (x-  l_n q^{-1/2} )  g(qx) +U(x) g(x), \label{eq:qsecond} \\
& U(x) =  -(q^{1/2} +q^{-1/2} ) x^2 + \sum _{n=1}^4 (h_n+ l_n ) x \nonumber \\
& \qquad + \prod_{n=1}^4 h_n^{1/2}  l_n^{1/2} \cdot  [ - (q^{1/2} +q^{-1/2}  ) x^{-2} + \sum _{n=1}^4 ( h_n^{-1} + l_n^{-1} ) x^{-1} ] . \nonumber
\end{align}
The equation $A^{\langle 2 \rangle} (x) g(x) =Eg(x) $ is also obtained as a specialization of the linear difference equation which is related with the $q$-Painlev\'e equation of type $E_7^{(1)}$ in \cite{Y}.
We discuss it in section \ref{sec:E7}. 

\subsection{Third degeneration}
\begin{prop} \label{prop:thirddeg}
In Eq.(\ref{eq:N1E7}), we replace $z$ by $z-iR$, $h_n $ $(n=1,2)$ by $h_n -iR$, $h_n $ $(n=3,4)$ by $h_n +iR$, $l_n $ $(n=1,2,3,4)$ by $l_n -iR$ and take the limit $R\to +\infty $.
Then we arrive at the operator
\begin{align}
A^{\langle 3 \rangle} (h,l;z) = V^{\langle 3 \rangle}(h;z)\exp(-ia_{-}\partial_{z})+W^{\langle 3 \rangle}(l;z)\exp(ia_{-}\partial_{z}) +U^{\langle 3 \rangle}(h,l;z),
\end{align}
where
\begin{align}
& V^{\langle 3 \rangle}(h;z) = e^{4\pi i  z} \prod_{n=1}^2 (1-e^{2\pi i h_n } e^{-\pi a_{-}}e^{-2\pi i  z}) , \label{eq:VWU3} \\
& W^{\langle 3 \rangle}(l;z) =e^{4\pi i z} \prod_{n=1}^4 (1- e^{2\pi i l_n} e^{\pi a_{-}} e^{-2\pi i z} ) , \nonumber \\
& U^{\langle 3 \rangle} (h,l;z) =\Big( \sum _{n=1}^2 e^{2\pi i h_n }+\sum _{n=1}^4 e^{2\pi i l_n } \Big)  e^{2\pi i z} -(e^{\pi a_- }+e^{-\pi a_- }) e^{4\pi i z} \nonumber \\
& \qquad +  e^{\pi i h_1}  e^{\pi i h_2}  ( e^{\pi i (h_3 - h_4 )} + e^{\pi i (h_4 - h_3 )})  \prod_{n=1}^4 e^{\pi i l_n} \cdot e^{-2\pi i z} . \nonumber
\end{align}
Namely, as $R\to +\infty $ we have
\begin{equation}
e^{-4\pi R} \tilde{A}^{\langle 2 \rangle} (h -iR v_1, l-iR v_2 ;z -iR) f(z) \to A^{\langle 3 \rangle} (h,l;z) f(z) 
\end{equation}
for any $f(z)$, where $v_1 = (1,1,-1,-1)$, $v_2=(1,1,1,1)$.
\end{prop}
\begin{proof}
We have 
\begin{align}
& \tilde{V}^{\langle 2 \rangle} (h -iRv_1 ;z -iR ) \\
& = e^{4\pi i z}e^{4\pi R} \prod_{n=1}^2 (1-e^{-2\pi i z}e^{2\pi i h_n } e^{-\pi a_{-}}) \prod_{n=3}^4 (1-e^{-4\pi R}e^{-2\pi i z}e^{2\pi i h_n } e^{-\pi a_{-}}) \nonumber \\
& \sim e^{4\pi R} e^{4\pi i z} \prod_{n=1}^2 (1-e^{-2\pi i z}e^{2\pi i h_n } e^{-\pi a_{-}}) , \nonumber
\end{align}
\begin{align}
& \tilde{W}^{\langle 2 \rangle} (l -iRv_2 ;z -iR )  =  e^{4\pi i z}e^{4\pi R} \prod_{n=1}^4 ( 1 -e^{-2\pi i z} e^{2\pi i l_n }e^{\pi a_{-}} ) , 
\end{align}
\begin{align}
& \tilde{U}^{\langle 2 \rangle} (h-iRv_1, l -iRv_2 ;z -iR )  =  -\Big( \sum _{n=1}^2 e^{2\pi i h_n }e^{2\pi R }+\sum _{n=3}^4 e^{2\pi i h_n }e^{-2\pi R }\\
& +\sum _{n=1}^4 e^{2\pi i l_n }e^{2\pi R } \Big) e^{2\pi i z}e^{2\pi R}  +(e^{\pi a_- }+e^{-\pi a_- })e^{4\pi i z}e^{4\pi R}  +e^{4\pi R } \prod _{n=1}^4 e^{\pi i h_n } \prod_{n=5}^8 e^{\pi i l_n } \Big\{  \nonumber \\
&  \Big( \sum _{n=1}^2  e^{-2\pi i h_n }e^{-2\pi R } +\sum _{n=3}^4 e^{-2\pi i h_n }e^{2\pi R } + \sum _{n=1}^4 e^{-2\pi i l_n }e^{-2\pi R } \Big) e^{-2\pi i z}e^{-2\pi R}\nonumber \\
&  - (e^{\pi a_- }+e^{-\pi a_- } ) e^{-4\pi i z}e^{-4\pi R} \Big\} \sim e^{4\pi R}  \Big\{ -\Big(\sum _{n=1}^2 e^{2\pi i h_n }+\sum _{n=1}^4 e^{2\pi i l_n } \Big) e^{2\pi i z} \nonumber \\
& +(e^{\pi a_- }+e^{-\pi a_- })e^{4\pi i z} +\Big( \prod _{n=1}^4 e^{\pi i h_n } e^{\pi i l_n } \cdot \sum _{n=3}^4 e^{-2\pi i h_n } \Big) e^{-2\pi i z}  \Big\} .\nonumber
\end{align}
Thus the proposition is obtained.
\end{proof}

Set $x= e^{2\pi i z}$, $q= e^{-2\pi a_{-}}$, and replace $e^{2\pi i h_n }$ and $e^{2\pi i l_n }$ by $h_n$ and $l_n$.
Then the difference operator $A^{\langle 3 \rangle} (h,l;z)$ is written as
\begin{align}
& A^{\langle 3 \rangle} (x) g(x) = \prod_{n=1}^2 (x-h_n q^{1/2} )  g(x/q) + x^{-2} \prod_{n=1}^4 (x- l_n q^{-1/2}  ) g(qx) +U(x) g(x),  \\
& U(x) = \Big (\sum _{n=1}^2 h_n +\sum _{n=1}^4 l_n \Big) x -(q^{1/2} +q^{-1/2} ) x^2 +  (l_1 l_2 l_3 l_4 h_1 h_2 )^{1/2} ( h_3^{1/2}h_4^{-1/2}+ h_3^{-1/2}h_4^{1/2}) x^{-1} .\nonumber
\end{align}
Let $E$ be a constant. The equation $A^{\langle 3 \rangle} (x) g(x) =Eg(x) $ is written as
\begin{align}
&  \prod_{n=1}^2 (x-h_n q^{1/2} ) g(x/q) + x^{-2} \prod_{n=1}^4 (x- l_n q^{-1/2}  ) g(qx) +\{ -(q^{1/2} +q^{-1/2} ) x^2  \label{eq:qthird} \\
&   +\Big( \sum _{n=1}^2 h_n +\sum _{n=1}^4 l_n \Big) x -E + (l_1 l_2 l_3 l_4 h_1 h_2 )^{1/2} ( h_3^{1/2}h_4^{-1/2}+ h_3^{-1/2}h_4^{1/2})  x^{-1} \} g(x) =0. \nonumber
\end{align}
This equation is also obtained as a specialization of the linear difference equation which is related with the $q$-Painlev\'e equation of type $E_6^{(1)}$ in \cite{Y}.
We discuss it in section \ref{sec:E6}. 
Set
\begin{align}
& \bar{A}^{\langle 3 \rangle} (x)  = v(x) \circ A^{\langle 3 \rangle} (h,l;z) \circ v(x)^{-1} ,\\
& v(x)= ( x^{-1} l_4 q^{1/2} ; q)_{\infty } =\prod _{k=0}^{\infty } (1-  x^{-1} l_4 q^{1/2} q^k) . \nonumber
\end{align}
We replace $h_3$, $h_4 $ and $l_4$ by $\tilde{h}$, $1$ and $h_3$ in Eq.(\ref{eq:qthird}). Then the gauge transformed equation $\bar{A}^{\langle 3 \rangle} (x) g(x) =Eg(x) $ is written as
\begin{align}
& \prod_{n=1}^3 (x-h_n q^{1/2} ) g(x/q) + \prod_{n=1}^3 (x- l_n q^{-1/2}  ) g(qx) +\Big\{ -(q^{1/2} +q^{-1/2} ) x^3  \label{eq:qthird''''} \\
&   +\sum _{n=1}^3 ( h_n + l_n ) x^2 -E x + (l_1 l_2 l_3 h_1 h_2 h_3)^{1/2} ( \tilde{h}^{1/2} + \tilde{h} ^{-1/2} ) \Big\} g(x) =0. \nonumber
\end{align}

To obtain the fourth degeneration, we apply the gauge transformation 
\begin{equation}
\tilde{A}^{\langle 3 \rangle} (h,l;z)  = R_- (z)^{2} \circ A^{\langle 3 \rangle} (h,l;z) \circ  R_- (z)^{-2} .
\end{equation}
Then we have
\begin{equation}
\tilde{A}^{\langle 3 \rangle} (h,l;z) = \tilde{V}^{\langle 3 \rangle} (h;z)\exp(-ia_{-}\partial_{z})+ \tilde{W}^{\langle 3 \rangle} (l;z)\exp(ia_{-}\partial_{z}) +U^{\langle 3 \rangle} (h,l;z),
\label{eq:N1E6}
\end{equation}
where $U^{\langle 3 \rangle} (h,l;z)$ was defined in Eq.(\ref{eq:VWU3}) and
\begin{align}
& \tilde{V}^{\langle 3 \rangle} (h;z) = e^{-2\pi a_{-}} \prod_{n=1}^2 (1-e^{2\pi i h_n } e^{-\pi a_{-}}e^{-2\pi i  z}),\\
& \tilde{W}^{\langle 3 \rangle} (l;z) = e^{-2\pi a_{-}} e^{8\pi i z}\prod_{n=1}^4 (1- e^{2\pi i l_n} e^{\pi a_{-}} e^{-2\pi i z} ) .\nonumber 
\end{align}
\subsection{Fourth degeneration and $q$-Heun equation}
\begin{prop} \label{prop:fourthdeg}
In Eq.(\ref{eq:N1E6}), we replace $z$ by $z+iR$, $h_n $ $(n=1,2,3,4)$ by $h_n +iR$, $l_n $ $(n=1,2)$ by $l_n +iR$, $l_n $ $(n=3,4)$ by $l_n -iR$  and take the limit $R\to +\infty $.
Then we arrive at the operator
\begin{align}
A^{\langle 4 \rangle} (h,l;z) = V^{\langle 4 \rangle}(h;z)\exp(-ia_{-}\partial_{z})+W^{\langle 4 \rangle}(l;z)\exp(ia_{-}\partial_{z}) +U^{\langle 4 \rangle}(h,l;z),
\label{eq:N1D50}
\end{align}
where
\begin{align}
& V^{\langle 4 \rangle} (h;z) = e^{-2\pi a_{-}} \prod_{n=1}^2 (1-e^{2\pi i h_n } e^{-\pi a_{-}}e^{-2\pi i  z}) , \label{eq:VWU4} \\
& W^{\langle 4 \rangle} (l;z) =  e^{4\pi i z} e^{2\pi i (l_3 +l_4)}  \prod_{n=1}^2 (1- e^{2\pi i l_n} e^{\pi a_{-}} e^{-2\pi i z} ) ,\nonumber \\
& U^{\langle 4 \rangle}(h,l;z) = (e^{2\pi i l_3 } + e^{2\pi i l_4 } ) e^{2\pi i z} +  e^{\pi i h_1}  e^{\pi i h_2}  ( e^{\pi i (h_3 - h_4 )} + e^{\pi i (h_4 - h_3 )})  \prod_{n=1}^4 e^{\pi i l_n} \cdot e^{-2\pi i z} .\nonumber 
\end{align}
Namely, as $R\to +\infty $ we have
\begin{equation}
e^{-4\pi R} \tilde{A}^{\langle 3 \rangle} (h +iR v_2, l+iR v_1 ;z +iR) f(z) \to A^{\langle 4 \rangle} (h,l;z) f(z) 
\end{equation}
for any $f(z)$, where $v_1 = (1,1,-1,-1)$, $v_2=(1,1,1,1)$.
\end{prop}
\begin{proof}
We have
\begin{align}
& \tilde{V}^{\langle 3 \rangle} (h +iRv_2 ;z +iR ) =  e^{-2\pi a_{-}} \prod_{n=1}^2 (1-e^{2\pi i h_n } e^{-\pi a_{-}}e^{-2\pi i  z}),
\end{align}
\begin{align}
& \tilde{W}^{\langle 3 \rangle} (l +iRv_1 ;z +iR ) = e^{-8\pi R} e^{-2\pi a_{-}} e^{8\pi i z}\prod_{n=1}^2 (1- e^{2\pi i l_n} e^{\pi a_{-}} e^{-2\pi i z} )\cdot \\
& \cdot  \prod_{n=3}^4 (1- e^{2\pi i l_n} e^{\pi a_{-}} e^{-2\pi i z} e^{4\pi  R}) \sim e^{4\pi i z}\prod_{n=1}^2 (1- e^{2\pi i l_n} e^{\pi a_{-}} e^{-2\pi i z} ) \prod_{n=3}^4 e^{2\pi i l_n} , \nonumber 
\end{align}
\begin{align}
& \tilde{U}^{\langle 3 \rangle} (h +iRv_2,l+ iRv_1 ;z +iR ) \\
& =  \Big( \sum _{n=1}^2 e^{2\pi i h_n } e^{-2\pi  R}+\sum _{n=1}^2 e^{2\pi i l_n }e^{-2\pi  R} +\sum _{n=3}^4 e^{2\pi i l_n }e^{2\pi  R}\Big) e^{2\pi i z}e^{-2\pi  R} \nonumber \\
& -(e^{\pi a_- }+e^{-\pi a_- }) e^{4\pi i z} e^{-4\pi  R}+ \prod_{n=1}^4 e^{\pi i l_n} \cdot e^{\pi i h_1}  e^{\pi i h_2} e^{-2\pi  R} ( e^{\pi i (h_3 - h_4 )} + e^{\pi i (h_4 - h_3 )})  e^{-2\pi i z} e^{2\pi  R} \nonumber \\
& \sim ( e^{2\pi i l_3 } + e^{2\pi i l_4 } ) e^{2\pi i z} + \prod_{n=1}^4 e^{\pi i l_n} \cdot e^{\pi i h_1}  e^{\pi i h_2}  ( e^{\pi i (h_3 - h_4 )} + e^{\pi i (h_4 - h_3 )}) e^{-2\pi i z} .\nonumber
\end{align}
\end{proof}
We apply a gauge transformation to the operator in Proposition \ref{prop:firstdeg} by using the function $R_-(z)$ defined in Eq.(\ref{eq:R+-}), which satisfies $R_- (z \mp ia_- ) =-e^{\pi a_-} e^{\pm 2\pi i z} R_- (z)$. 

By the multiplication and the gauge transformation given by
\begin{equation}
\tilde{A}^{\langle 4 \rangle} (h,l;z)  = - R_- (z)^{-1}  e^{- \pi i z }\circ A^{\langle 4 \rangle} (h,l;z) \circ  R_- (z) e^{\pi i z } , 
\end{equation}
we have 
\begin{align}
\tilde{A}^{\langle 4 \rangle} (h,l;z) = \tilde{V}^{\langle 4 \rangle} (h;z)\exp(-ia_{-}\partial_{z})+\tilde{W}^{\langle 4 \rangle}_{j}(l;z)\exp(ia_{-}\partial_{z}) -U^{\langle 4 \rangle}(h,l;z),
\label{eq:N1D5}
\end{align}
where was defined in Eq.(\ref{eq:VWU4}) and 
\begin{align}
& \tilde{V}^{\langle 4 \rangle} (h;z) = e^{2\pi i z} \prod_{n=1}^2 (1-e^{2\pi i h_n } e^{-\pi a_{-}}e^{-2\pi i  z}) ,\\
& \tilde{W}^{\langle 4 \rangle} (l;z) = e^{2\pi i l_3} e^{2\pi i l_4}  e^{2 \pi i z} \prod_{n=1}^2 (1- e^{2\pi i l_n} e^{\pi a_{-}} e^{-2\pi i z} ) .\nonumber 
\end{align}
Set $x= e^{2\pi i z}$, $q= e^{-2\pi a_{-}}$, replace $e^{2\pi i h_n }$ and $e^{2\pi i l_n }$ by $h_n$ and $l_n$, and set $h_4=1$.
Then the difference operator is written as
\begin{align}
& A^{\langle 4 \rangle} (x) g(x) =  x^{-1} (x-h_1 q^{1/2} ) (x-h_2 q^{1/2} )  g(x/q) + x^{-1}  l_3 l_4 (x - l_1 q^{-1/2} ) (x - l_2 q^{-1/2} ) g(qx) \\
& \qquad -\{ (l_3 +l_4 ) x  + (l_1 l_2 l_3 l_4 h_1 h_2 )^{1/2} ( h_3^{1/2}+ h_3^{-1/2}) x^{-1} \} g(x) . \nonumber
\end{align}
Let $E$ be a constant. The equation $A^{\langle 4 \rangle} (x) g(x) =Eg(x) $ is written as
\begin{align}
& (x-h_1 q^{1/2} ) (x-h_2 q^{1/2} ) g(x/q)  + l_3 l_4 (x - l_1 q^{-1/2} ) (x - l_2 q^{-1/2} ) g(xq) \label{eq:RuijD5}
\\
&  \qquad -\{ (l_3 +l_4 ) x^2 + E x + (l_1 l_2 l_3 l_4 h_1 h_2 )^{1/2} ( h_3^{1/2}+ h_3^{-1/2}) \} g(x) =0. \nonumber
\end{align}
This equation is also obtained as a specialization of the linear difference equation which is related with the $q$-Painlev\'e VI equation \cite{JS}.
We discuss it in section \ref{sec:D5}.

We call Eq.(\ref{eq:RuijD5}) the $q$-Heun equation, because it has a limit to the Heun equation, which we show in the rest of this subsection.
We rewrite Eq.(\ref{eq:RuijD5}) as
\begin{align}
& (x-t_1q^{h_1} q^{1/2} ) (x-t_2 q^{h_2} q^{1/2} ) g(x/q)  + q^{l_3 +l_4} (x - t_1 q^{l_1} q^{-1/2} ) (x - t_2 q^{l_2} q^{-1/2} ) g(xq)] \label{eq:qHeun} \\
& \qquad  -\{ (q^{l_3} +q^{l_4} ) x^2 - (2 (t_1+t_2) +(q-1)E_1 +(q-1)^2 \tilde{E} ) x \nonumber \\
& \qquad + t_1 t_2 q^{(l_1 +l_2 +l_3 +l_4 +h_1 + h_2)/2} ( q^{h_3/2} + q^{-h_3/2}) \} g(x) =0,\nonumber
\end{align}
where
\begin{equation}
E_1= (l_3 +l_4 )(t_1+t_2) +(l_1 +h_1 )t_1 +(l_2 +h_2 )t_2.
\end{equation}
Set $q=1+ \varepsilon $.
We divide Eq.(\ref{eq:qHeun}) by $\varepsilon ^2$.
By using Taylor's expansion
\begin{align}
& g(x/q) =g(x) + (-\varepsilon +\varepsilon ^2 )xg'(x) + \varepsilon ^2 x^2 g''(x) /2 +O( \varepsilon ^3),\\
& g(qx) =g(x) + \varepsilon  x g'(x) + \varepsilon ^2 x^2 g''(x) /2 +O( \varepsilon ^3), \nonumber 
\end{align}
we find the following limit as $\varepsilon  \to 0$:
\begin{align}
& x^2(x-t_1)(x-t_2) g''(x) \label{eq:Fuchs4sing} \\
& +  [ (1+ h_2 - l_2)x(x- t_1 ) + (1 +h_1 - l_1 )x(x- t_2) + ( \tilde{l}  -1 ) (x -t_1)(x-t_2) ]  xg'(x) \nonumber \\
& + [ ( l_3 l_4 e^2  ) x^2 + \tilde{B} x + t_1 t_2 ( \tilde{l}/2 -1 +h_3/2)( \tilde{l}/2 -1 -h_3/2)] g(x) =0, \nonumber 
\end{align}
where
\begin{align}
& \tilde{l} = l_1 +l_2 +l_3 +l_4 -h_1 - h_2 , \\ 
& \tilde{B} = \tilde{E} - \frac{t_1}{2} \Big\{ h_1^2 + (l_3 +l_4+ l_1 - 1)^2 - \frac{1}{2} \Big\} + \frac{t_2}{2} \Big\{ h_2^2 + (l_3 +l_4+ l_2 - 1)^2 -\frac{1}{2} \Big\} . \nonumber 
\end{align}
This equation is a Fuchsian differential equation with four singularities $\{ 0,t_1, t_2, \infty \}$ and the local exponents are given by the following Riemann scheme:
\begin{equation}
\begin{pmatrix}
x=0 &  x=t_1 &  x=t_2 &  x=\infty  \\
1-\tilde{l}/2 +h_3/2 & 0 & 0 & l_3 \\ 
1-\tilde{l}/2 -h_3/2 & l_1-h_1 & l_2-h_2 & l_4  
\end{pmatrix}
\end{equation}
By setting $z=x/t_1$ and $g(x)=x^{1-\tilde{l}/2 - h_3/2 }\tilde{g}(z)$, the function $y=\tilde{g}(z) $ satisfies Heun's differential equation;
\begin{align}
\frac{d^2y}{dz^2} +  \left( \frac{\gamma}{z}+\frac{\delta }{z-1}+\frac{\epsilon}{z-t}\right) \frac{dy}{dz} + \frac{\alpha \beta z -q}{z(z - 1)(z - t)} y= 0,
\label{eq:Heun1}
\end{align}
where $t=t_2/t_1$, $\gamma = 1 - h_3$, $\delta = 1 +h_1 -l_1$, $\epsilon = 1+h_2-l_2$, $\{ \alpha ,\beta \} = \{ l_3 +1 -\tilde{l}/2 -h_3/2 , l_4 + 1 -\tilde{l}/2 -h_3/2 \} $ and $t=t_2/t_1$.
Therefore it is reasonable to call Eq.(\ref{eq:RuijD5}) the $q$-Heun equation.

\section{Linear $q$-difference equations related with $q$-Painlev\'e equations} \label{sec:JSY}
It is widely known that the Lax formalism is a powerful method for studying soliton equations, and a similar method is applied for Painlev\'e-type equations.
In particular Jimbo and Sakai \cite{JS} obtained $q$-Painlev\'e VI (or the $q$-Painlev\'e equation of type $D^{(1)}_5$) by finding "Lax forms".
We may regard the Lax forms for difference Painlev\'e equations as a linear difference equation and an associated deformation equation where the difference Painlev\'e equation is written in the form of compatibility of them (see \cite{SakP}). 

About ten years later from the discovery of $q$-Painlev\'e VI, Sakai \cite{SakP} presented a problem to find Lax forms for difference Painlev\'e equations in his list \cite{Sak}, and Yamada made a significant contribution for the problem.
Namely, Yamada discovered Lax forms for the $q$-difference Painlev\'e equations of types $D_5^{(1)}$, $E_6^{(1)}$, $E_7^{(1)}$, $E_8^{(1)}$ \cite{Y} and the elliptic difference Painlev\'e equation \cite{Ye}.
For type $D_5^{(1)}$, it essentially coincides with the one by Jimbo and Sakai \cite{JS}.

In this section, we observe that our degenerete operators of Ruijsenaars-van Diejen appear by restricting the parameters in the linear $q$-differential equations related to $q$-Painlev\'e equations of types $D_5^{(1)}$, $E_6^{(1)}$, $E_7^{(1)}$.

\subsection{Linear $q$-difference equation related with $q$-Painlev\'e VI} \label{sec:D5}

Jimbo and Sakai \cite{JS} found a Lax form for the $q$-Painlev\'e VI by considering connection preserving deformation of the linear system of $q$-differential equations
\begin{align}
& Y(qx,t)= A(x,t) Y(x,t) . \label{eq:JSlinDEqx}
\end{align}
To derive $q$-Painlev\'e VI, they rewrite the condition for connection preserving deformation as compatibility of Eq.(\ref{eq:JSlinDEqx}) with a deformation equation of the form $Y(x,qt)= B(x,t) Y(x,t)$, which is a Lax form of $q$-Painlev\'e VI.

We focus on Eq.(\ref{eq:JSlinDEqx}).
The $2\times 2$ matrix $A(x,t)$ is taken in the form Eqs.(9)-(11) in \cite{JS}, i.e.
\begin{align}
& A(x,t)= A_0 (t) +A_1 (t) x +A_2 x^2 = 
\begin{pmatrix}
a_{11}(x) & a_{12}(x) \\
a_{21}(x) & a_{22}(x)
\end{pmatrix} ,\\
& A_2 = \begin{pmatrix}
\kappa _1 & 0 \\
0 & \kappa _2
\end{pmatrix} , \quad A_0(t) \mbox{ has eigenvalues } t\theta _1, t\theta _2 , \nonumber  \\
& \det A(x,t)= \kappa _1 \kappa _2 (x-ta_1)(x-ta_2)(x-a_3)(x-a_4). \nonumber
\end{align}
Note that we have the relation $\kappa_ 1 \kappa_2 a_1 a_2 a_3 a_4 =\theta _1 \theta _2$.
Define $\lambda $, $\mu _1$, $\mu _2$ by 
\begin{align}
& a_{12}(\lambda )=0, \; \mu _1= a_{11}(\lambda )/\kappa _1 , \; \mu _2 = a_{22}(\lambda )/ \kappa _2
\end{align}
so that $\mu _1 \mu _2 =(\lambda -ta_1)(\lambda -ta_2)(\lambda -a_3)(\lambda -a_4)$ (see Eq.(14) in \cite{JS}) and introduce $\mu $ by $\mu =(\lambda -ta_1)(\lambda -ta_2)/(q\kappa _1 \mu _1)$.
Then the matrix elements can be parametrized by these variables and the gauge freedom $w$ (see page 149 of \cite{JS}).

The first component of $Y(x)$ satisfies the following equation:
\begin{align}
& Y_1 (q^2x) - \Big( a_{11}(qx) + \frac {a_{12} (qx)}{a_{12}(x)} a_{22} (x) \Big) Y_1 (qx) \label{eq:qY1} \\
& \qquad + \frac {a_{12} (qx)}{a_{12}(x)} (a_{11}(x) a_{22} (x) - a_{12}(x) a_{21} (x) ) Y_1 (x)=0 . \nonumber
\end{align}
In our parametrisation, we have
\begin{align}
& \frac {a_{12} (qx)}{a_{12}(x)} (a_{11}(x) a_{22} (x) - a_{12}(x) a_{21} (x) ) = \frac{qx-\lambda  }{x-\lambda  } \kappa _1 \kappa _2 (x-ta_1)(x-ta_2)(x-a_3)(x-a_4) , \label{eq:qY1coeff} \\
& a_{11}(qx) +\frac {a_{12} (qx)}{a_{12}(x)} a_{22} (x)= \frac{q (q \kappa_ 1 + \kappa _2 ) x^3+c_2 x^2 + c_1 x - \lambda   t (\theta _1 +\theta _2 )}{x-\lambda  } , \nonumber\\
& c_2 = \frac{q^2 \kappa_ 1 \kappa _2 (\lambda  -a_3) (\lambda  -a_4) \mu  }{\lambda }  -(q+1) (q \kappa_ 1 + \kappa _2) \lambda  -\frac{q t (\theta _1 +\theta _2 ) }{\lambda } +\frac{(\lambda  -a_1 t) (\lambda  -a_2 t)}{\lambda  \mu } , \nonumber \\
& c_1 = -q \kappa_ 1 \kappa _2 (\lambda  -a_3) (\lambda  -a_4) \mu +(q\kappa_ 1 + \kappa _2 ) \lambda  ^2+(q+1) t (\theta _1 +\theta _2 )-\frac{(\lambda  -a_1 t) (\lambda  -a_2 t)}{\mu } . \nonumber
\end{align}
Note that there are two accessory parameters $\lambda $ and $\mu $, which play the role of dependent variables in the $q$-Painlev\'e VI equation (see Eqs.(19), (20) in \cite{JS}).
We may regard Eq.(\ref{eq:qY1}) with Eq.(\ref{eq:qY1coeff}) as a $q$-difference analogue of Eq.(\ref{eq:linP6}) in the differential equations.
We impose a restriction on the accessory parameters as we have done for Eq.(\ref{eq:linP6}) to obtain the Heun equation.
Here we require $\lambda =a_3$. Then we have
\begin{align}
& Y_1 (q^2x) - \{ q (q \kappa_ 1 + \kappa _2 ) x^2  +d _1 x + t (\theta _1 +\theta _2 ) \} Y_1 (qx) \\
& \qquad + \kappa _1 \kappa _2 (qx-a_3) (x-ta_1)(x-ta_2)(x-a_4) Y_1 (x)=0 , \nonumber 
\end{align}
where
\begin{align}
& d_1 = \frac{(a_1 t-a_3)(a_2 t-a_ 3)}{a_3 \mu } -a_3 (q \kappa_ 1 + \kappa _2 ) - \frac{q t (\theta _1 +\theta _2 )}{a_3} . 
\end{align}
Let $u(x)$ be the function which satisfies $u(qx) =(x-ta_1)(x-ta_2) u(x) $, e.g. $u(x)=x^{(\log t^2a_1 a_2)/\log q}(x/(ta_1);q)_{\infty } (x/(ta_2);q)_{\infty } $.
Then the function $f(x)=Y_1 (qx)/u(qx)$ satisfies
\begin{align}
& (x-ta_1)(x-ta_2) f (qx) - \{ ((q \kappa_ 1 + \kappa _2 )/q) x^2 +(d_1/q) x  + t (\theta _1 +\theta _2 ) \} f(x) \label{eq:qlinP6} \\
& \qquad + (\kappa _1 \kappa _2 /q) (x-a_3) (x-q a_4) f(x/q) =0.\nonumber 
\end{align}
Note that there is the relation $\kappa_ 1 \kappa_2 a_1 a_2 a_3 a_4 =\theta _1 \theta _2$.
In Eq.(\ref{eq:RuijD5}), we set
\begin{align}
& l_1= a_1 t q^{1/2}, \; l_2= a_2 t q^{1/2}, \; h_1 = a_3 q^{-1/2}, \; h_2 = a_4 q^{1/2},\\
& l_3= 1/\kappa _1 , \; l_4 = q/ \kappa _2 , \; h_3 ^{1/2}= \theta _1 (a_1 a_2 a_3 a_4 \kappa _1 \kappa _2  )^{-1/2}, \; E =d_1/(\kappa _1 \kappa _2). \nonumber 
\end{align}
Then we have $h_3 ^{-1/2}= \theta _2 (a_1 a_2 a_3 a_4 \kappa _1 \kappa _2  )^{-1/2}  $ and Eq.(\ref{eq:qlinP6}).

Hence Eq.(\ref{eq:qlinP6}) is obtained by the fourth degeneration of the Ruijsenaars-van Diejen operator.
We may regard $d_1$ as an accessory parameter.

\subsection{Linear $q$-difference equation related with $q$-Painlev\'e equation of type $E_6^{(1)}$} \label{sec:E6}
Yamada \cite{Y} derived a $q$-difference Painlev\'e equation of type $E_6^{(1)}$ by Lax formalism, i.e. the compatibility condition for two linear $q$-difference equations.
One of the linear difference equations is written as Eq.(40) in \cite{Y}, i.e.
\begin{align}
& \frac{(b_1q-z)(b_2q-z)(b_3q-z)(b_4q-z)t^2}{q(fq-z) z^4}\left[ y(z/q) -\frac{gz}{t^2(gz-q)}y(z) \right] \label{eq:E6lin} \\
& + \frac{(b_5 t-z)(b_6 t-z)}{(f-z) z^2 t^2}\left[ y(qz) -\frac{(gz-1)t^2}{gz}y(z) \right] \nonumber \\
& +\left[ \frac{(b_1 g-1)(b_2 g-1)(b_3 g-1)(b_4 g-1)t^2}{g(fg-1) z^2(gz-q)} - \frac{b_5 b_6 (b_7 g-t) (b_8g-t) }{fg z^3} \right] y(z)=0.\nonumber 
\end{align}
We may regard $f$, $g$ as accessory parameters.
The other linear equation (deformation equation) contains the difference on the variable $t$ as well as the variable $x$ (see Eq.(40) in \cite{Y}).
By compatibility condition for two linear difference equations, we have $q$-Painlev\'e equation of type $E_6^{(1)}$ for the dependent variables $f$ and $g$ and independent variable $t$ (see Eq.(39) in \cite{Y}).

We now specialize the parameters $f$ and $g$ to $f=b_1 $ in Eq.(\ref{eq:E6lin}).
Then we obtain
\begin{align}
& \frac{(b_2q -z)(b_3q -z)(b_4q -z)t^2}{q z^4}y(z/q) + \frac{c(z)}{z^2 q^{1/2} (b_1-z)}y(z) + \frac{(z-b_5 t)(z-b_6 t)}{(b_1-z)z^2t^2} y(qz)=0, \label{eq:E6linsp0} \\
& c(z) = -(q^{1/2 }+q^{-1/2}) z^2 + (b_1q^{-1/2}+b_2 q^{1/2}+b_ 3 q^{1/2}+b_4 q^{1/2}+ b_5 tq^{1/2}+ b_6 tq^{1/2})z \nonumber \\
& \qquad +c_1 + \frac{q^{1/2} t b_5 b_6 (b_7 + b_8 )}{ z } , \nonumber
\end{align}
and the term $c_1$ contains the accessory parameter $g$.
Note that there is a relation $q b_1 b_2 b_3 b_4 = b_5 b_6 b_7 b_8$ in Yamada's paper.
We apply a gauge transformation by $y(z)= z^{-1/2+2(\log t)/\log q} \tilde{y}(z)$.
Then we have
\begin{align}
& \frac{(z-b_1)(z- b_2q)(z- b_3q)(z- b_4q)}{ z^2}\tilde{y}(z/q)+c(z) \tilde{y}(z) + (z-b_5 t)(z-b_6 t) \tilde{y}(qz)=0. \label{eq:E6linsp}
\end{align}
By replacing $q$ to $q^{-1}$, it turns out that Eq.(\ref{eq:E6linsp}) is the third degeneration of the Ruijsenaars-van Diejen operator in Eq.(\ref{eq:qthird}).
The eigenvalue $E$ in Eq.(\ref{eq:qthird}) essentially corresponds to the accessory parameter $g$ in $c_1$.

\subsection{Linear $q$-difference equation related with $q$-Painlev\'e equation of type $E_7^{(1)}$} \label{sec:E7}
Yamada \cite{Y} also derived a $q$-difference Painlev\'e equation of type $E_7^{(1)}$ by Lax formalism, i.e. the compatibility condition for two linear $q$-difference equations.
Set
\begin{equation}
B_1(z)=(1-b_1 z) (1-b_2 z) (1-b_3 z) (1-b_4 z), \; B_2(z)=(1-b_5 z) (1-b_6 z) (1-b_7 z) (1-b_8 z).
\end{equation}
Then one of the $q$-difference equations is written as Eq.(37) in \cite{Y}, i.e. 
\begin{align}
& \frac{B_ 2(t/z)}{t^2 (f-z)} \left[ y(qz) - \frac{t^2(1-gz)}{t^2-gz} y(z) \right]  + \frac{t^2 B_1(q/z)}{q(f q-z)} \left[ y(z/q) - \frac{qt^2-gz}{t^2(q-gz)} y(z) \right] \label{eq:E7lin} \\
& + \frac{(1-t^2)}{g z^2} \left[ \frac{q B_1(g)}{(f g-1)(g z-q)}-\frac{t^4 B_2(g/t)}{(fg-t^2)(gz-t^2)}  \right] y(z)=0. \nonumber
\end{align}
We also obtain the $q$-Painlev\'e equation of type $E_7^{(1)}$ (see Eq.(36) in \cite{Y}) by a compatibility condition of Eq.(\ref{eq:E7lin}) with the deformation equation written as Eq.(38) in \cite{Y}.

By specializing to $f=b_1 $ in Eq.(\ref{eq:E7lin}) and applying a gauge transformation, we have 
\begin{align}
& (z-b_5 t) (z-b_6 t) (z-b_7 t) (z-b_8 t) y(qz) -c(z) y(z) \label{eq:E7linsp} \\
& \qquad +(z-b_1 ) (z-b_2 q ) (z-b_3 q ) (z-b_4 q ) y(z/q)=0, \nonumber
\end{align}
where
\begin{align}
& c(z) = q^{-1/2} \{(1+q) z^4 +c_3 z^3 +c_2 z^2 +c_1 z + (b_5 b_6 b_7 b_8+q^2 b_1 b_2 b_3 b_4) t^2 q   \} ,\\
& c_3= -(b_1+b_2 q+b_3 q+b_4 q+ b_5 t q+ b_6 t q+ b_7 t q +b_8 t q ) , \nonumber \\
& c_1= -q (b_2 b_3 b_4 t^2 q^2 +b_1 b_2 b_4 t^2 q  +b_1 b_3 b_4 t^2 q + b_1 b_2 b_3 t^2 q  \nonumber \\
& \qquad +b_6 b_7 b_8 t +b_5 b_7 b_8 t +b_5 b_6 b_8 t +b_5 b_6 b_7 t) ,\nonumber 
\end{align}
and the term $c_2$ contains the accessory parameter $g$.
There is a relation $q b_1 b_2 b_3 b_4 = b_5 b_6 b_7 b_8$ in Yamada's paper.
Eq.(\ref{eq:E7linsp}) is the second degeneration of Ruijsenaars-van Diejen operator in Eq.(\ref{eq:qsecond}) by setting
\begin{align}
& h_1 =b_1 q^{-1/2}, \; h_2 =b_2 q^{1/2}, \; h_3 =b_3 q^{1/2}, \; h_4 =b_4 q^{1/2},\\
& l_1=b_5 q^{1/2} t, \; l_2=b_6 q^{1/2} t, \; l_3=b_7 q^{1/2} t, \; l_4=b_8 q^{1/2} t . \nonumber
\end{align}
Note that we used the relation $(h_1 h_2 h_3 h_4 l_1 l_2 l_3 l_4)^{1/2} =  qt^2 b_5 b_6 b_7 b_8 = q^2t^2 b_1 b_2 b_3 b_4$, which follows from $q b_1 b_2 b_3 b_4 = b_5 b_6 b_7 b_8$.

\section{Degeneration of Ruijsenaars-van Diejen operator with $N$ variables} \label{sec:degen}

We describe Ruijsenaars-van Diejen operator (\ref{eq:RvDintro}) of $N$ variables and investigate degenerations of the opetator.
By using the notation in section \ref{sec:degenone}, the Ruijsenaars-van Diejen operator of $N$ variables is given by
\begin{equation}
A_{+}(h;z) = \sum _{j=1}^N ( V_{j,+}(h;z)\exp(-ia_{-}\partial_{z_j})+V_{j,+}(h;-z)\exp(ia_{-}\partial_{z_j}) )+V_{b,+}(h;z),
\label{eq:defeRvDN}
\end{equation}
where
\begin{align}
& V_{j,+}(h;z) = \frac{\prod_{n=1}^8 R_{+}(z_j-h_n-ia_{-}/2)}{R_{+}(2z_j+ia_{+}/2)R_{+}(2z_j-ia_{-}+ia_{+}/2)} \cdot \\
& \qquad \cdot \prod_{k\neq j} \frac{R_{+}(z_j-z_k -\mu +ia_{+}/2)R_{+}(z_j+z_k -\mu +ia_{+}/2)}{R_{+}(z_j-z_k +ia_{+}/2)R_{+}(z_j+z_k +ia_{+}/2)}, \nonumber \\
& V_{b,+}(h;z) = \frac{\sum_{t=0}^3p_{t,+}(h)
[\prod _{j=1}^N \cE_{t,+}(\mu ;z_j)-\cE_{t,+}(\mu ;\omega_{t,+})^N]}{2R_{+}(\mu -ia_{+}/2)R_{+}(\mu  -ia_{-}-ia_{+}/2)}.\nonumber 
\end{align}

We set $h_n= \tilde{h}_n -ia_+/2$.
By the limit $q_+ \to 0$, we have the following proposition, whose proof is similar to the case of one variable.
\begin{prop} \label{prop:firstdegN}
As $q_+ \to 0$, the Ruijsenaars-van Diejen operator in Eq.(\ref{eq:defeRvDN}) corresponds to the operator
\begin{equation}
A^{\langle 1 \rangle}  (h;z) = \sum _{j=1}^N ( V^{\langle 1 \rangle} _{j}(h;z)\exp(-ia_{-}\partial_{z_j})+V^{\langle 1 \rangle} _{j}(h;-z)\exp(ia_{-}\partial_{z_j}) )+U^{\langle 1 \rangle} (h;z),
\end{equation}
up to some additive constant, where 
\begin{align}
& V^{\langle 1 \rangle} _{j}(h;z) = \frac{\prod_{n=1}^8 (1-e^{-2\pi i z_j}e^{2\pi i \tilde{h}_n }e^{-\pi a_{-}})}{(1-e^{-4 \pi i z_j })(1-e^{-4\pi i z_j}e^{-2\pi a_{-}})} \prod_{k\neq j} \frac{(1-e^{2\pi i (z_k-z_j)} e^{2\pi i \mu })(1-e^{-2\pi i (z_k+z_j)}e^{2\pi i \mu }) }{(1-e^{2\pi i (z_k-z_j)})(1-e^{-2\pi i (z_k+z_j)}) },
\end{align}
\begin{align}
& U^{\langle 1 \rangle} (h;z) = \label{eq:Vbdeg1N} \\
&  \frac{\prod_{n=1}^8 (e^{2\pi i \tilde{h}_n }- 1)}{2(e^{2\pi i \mu  } -1)(e^{2\pi i \mu }e^{2 \pi a_- } -1) } \prod _{j=1}^N \Big\{ e^{2\pi i \mu  } + \frac{(e^{2\pi i \mu  } -1)(e^{2\pi i \mu }e^{2 \pi a_- } -1)}{(1-e^{2\pi i z_j}e^{\pi a_- })(1-e^{-2\pi i z_j}e^{\pi a_- })} \Big\} \nonumber \\
& + \frac{\prod_{n=1}^8 (e^{2\pi i \tilde{h}_n } + 1) }{2 (e^{2\pi i \mu  } -1)(e^{2\pi i \mu }e^{2 \pi a_- } -1)} \prod _{j=1}^N \Big\{ e^{2\pi i \mu  } + \frac{(e^{2\pi i \mu  } -1)(e^{2\pi i \mu }e^{2 \pi a_- } -1)}{(1+e^{2\pi i z_j}e^{\pi a_- })(1+e^{-2\pi i z_j}e^{\pi a_- })} \Big\} \nonumber \\
& + e^{-\pi a_- } e^{2(N-1)\pi i \mu }\prod_{n=1}^8 e^{\pi i \tilde{h}_n } \Big[ \sum _{n=1}^8  ( e^{2\pi i \tilde{h}_n }+ e^{-2\pi i \tilde{h}_n }) \sum _{j=1}^N  (e^{2\pi i z_j} +  e^{-2\pi i z_j} ) \nonumber \\
&  - (e^{-\pi a_- }+e^{\pi a_- } )\sum _{j=1}^N (e^{4\pi i z_j} +  e^{-4\pi i z_j} ) \nonumber \\
& +(e^{\pi i \mu  }  -e^{-\pi i \mu  }) (e^{\pi i \mu  } e^{\pi a_- } - e^{-\pi i \mu  }e^{-\pi a_- }) \sum _{1\leq j <k \leq N}  (e^{2\pi i z_j} +  e^{-2\pi i z_j} )  (e^{2\pi i z_k} +  e^{-2\pi i z_k} ) \Big] . \nonumber
\end{align}
\end{prop}
By applying the gauge transformation with respect to the function $(R_- (z_1) \dots R_- (z_N) )^2$, we arrive at the following operator
\begin{equation}
\tilde{A}^{\langle 1 \rangle}  (h;z) = \sum _{j=1}^N ( \tilde{V}^{\langle 1 \rangle} _{j}(h;z)\exp(-ia_{-}\partial_{z_j})+\tilde{W}^{\langle 1 \rangle} _{j}(h;z)\exp(ia_{-}\partial_{z_j}) )+U^{\langle 1 \rangle} (h;z),
\label{eq:NE8}
\end{equation}
where $U^{\langle 1 \rangle} (h;z) $ was defined in Eq.(\ref{eq:Vbdeg1N}) and 
\begin{align}
& \tilde{V}^{\langle 1 \rangle}_{j}(h;z) = \frac{\prod_{n=1}^8 (1-e^{-2\pi i z_j}e^{2\pi i \tilde{h}_n }e^{-\pi a_{-}})}{e^{-2\pi a_{-}} e^{-4 \pi i z_j }(1-e^{-4 \pi i z_j })(1-e^{-4\pi i z_j}e^{-2\pi a_{-}})} \cdot \\
& \quad \cdot \prod_{k\neq j} \frac{(1-e^{2\pi i (z_k- z_j)}e^{2\pi i \mu })(1-e^{-2\pi i (z_k+z_j )}e^{2\pi i \mu }) }{(1-e^{2\pi i (z_k- z_j)})(1-e^{-2\pi i (z_k+z_j )}) }, \nonumber \\
& \tilde{W}^{\langle 1 \rangle}_{j}(h;z) = \frac{\prod_{n=1}^8 (1-e^{2\pi i z_j}e^{2\pi i \tilde{h}_n }e^{-\pi a_{-}})}{e^{-2\pi a_{-}} e^{4 \pi i z_j }(1-e^{4 \pi i z_j })(1-e^{4\pi i z_j}e^{-2\pi a_{-}})} \cdot \nonumber \\
& \quad \cdot \prod_{k\neq j} \frac{(1-e^{-2\pi i (z_k -z_j)}e^{2\pi i \mu })(1-e^{2\pi i (z_k +z_j)}e^{2\pi i \mu }) }{(1-e^{-2\pi i (z_k -z_j)})(1-e^{2\pi i (z_k +z_j)}) }. \nonumber 
\end{align}
Note that this operator was essentially obtained by van Diejen \cite{vD}.

We apply the second degeneration.
\begin{prop} \label{prop:secdegN}
In Eq.(\ref{eq:NE8}), we replace $z$ by $z+iR$, $\tilde{h}_n $ $(n=1,2,3,4)$ by $h_n +iR$, $\tilde{h}_n $ $(n=5,6,7,8)$ by $h_n -iR$ and take the limit $R\to +\infty $.
Then we have the operator
\begin{equation}
A^{\langle 2 \rangle} (h;z) = \sum _{j=1}^N ( V^{\langle 2 \rangle}_{j}(h;z)\exp(-ia_{-}\partial_{z_j})+W^{\langle 2 \rangle}_{j}(h;z)\exp(ia_{-}\partial_{z_j}) )+U^{\langle 2 \rangle} (h;z),
\end{equation}
where 
\begin{align}
& V^{\langle 2 \rangle}_{j}(h;z) = e^{4\pi i  z_j} e^{2(N-1)\pi i \mu } \prod_{n=1}^4 (1-e^{-2\pi i  z_j}e^{2\pi i h_n } e^{-\pi a_{-}})\prod_{n=5}^8 e^{2\pi i h_n}  \prod_{k\neq j}  \frac{(1-e^{2\pi i (z_k -z_j)}e^{2\pi i \mu })}{(1-e^{2\pi i (z_k -z_j) })}, \\
& W^{\langle 2 \rangle}_{j}(h;z) = e^{2\pi a_{-}} e^{-4\pi i z_j} \prod_{n=5}^8 (1-e^{2\pi i z_j}e^{2\pi i h_n} e^{-\pi a_{-}})  \prod_{k\neq j}  \frac{(1-e^{-2\pi i (z_k -z_j)}e^{2\pi i \mu })}{(1-e^{-2\pi i (z_k -z_j)})}, \nonumber 
\end{align}
\begin{align}
& U^{\langle 2 \rangle} (h;z) =
 e^{2(N-1) \pi i \mu  }\prod_{n=5}^8 e^{2\pi i h_n } \Big[ \Big( \sum _{n=1}^4 e^{2\pi i h_n }+\sum _{n=5}^8 e^{-2\pi i h_n }\Big) e^{-\pi a_- } \sum _{j=1}^N  e^{2\pi i z_j} \nonumber \\
& -(1+e^{-2\pi a_- })\sum _{j=1}^N e^{4\pi i z_j}  +e^{-2 \pi i \mu  } e^{-2\pi a_- } (e^{2\pi i \mu  } -1)(e^{2\pi i \mu }e^{2 \pi a_- } -1) \sum _{1\leq j<k \leq N} e^{2\pi i z_j}e^{2\pi i z_k} \Big] \nonumber \\
& +  e^{2(N-1)\pi i \mu }\prod_{n=1}^8 e^{\pi i h_n } \Big[ \Big( \sum _{n=1}^4  e^{-2\pi i h_n } + \sum _{n=5}^8  e^{2\pi i h_n } \Big) e^{-\pi a_- } \sum _{j=1}^N  e^{-2\pi i z_j} \nonumber \\
&  - (1+e^{-2\pi a_- } )\sum _{j=1}^N  e^{-4\pi i z_j} +e^{-2 \pi i \mu  } e^{-2\pi a_- } (e^{2\pi i \mu  } -1)(e^{2\pi i \mu }e^{2 \pi a_- } -1)  \sum _{1\leq j <k \leq N}  e^{-2\pi i z_j} e^{-2\pi i z_k}  \Big] . \nonumber
\end{align}
\end{prop}

Set $l_{n}= -h_{n+4}$ $(n=1,2,3,4) $.
By multiplication and gauge transformation, we have
\begin{equation}
\tilde{A}^{\langle 2 \rangle} (h,l;z) = \sum _{j=1}^N ( \tilde{V}^{\langle 2 \rangle} _{j}(h;z)\exp(-ia_{-}\partial_{z_j})+\tilde{W}^{\langle 2 \rangle} _{j}(l;z)\exp(ia_{-}\partial_{z_j}) )+ \tilde{U}^{\langle 2 \rangle} (h,l;z),
\label{eq:NE7}
\end{equation}
where
\begin{align}
& \tilde{V}^{\langle 2 \rangle} _{j}(h;z) = e^{4\pi i  z_j} \prod_{n=1}^4 (1-e^{2\pi i h_n } e^{-\pi a_{-}}e^{-2\pi i  z_j}) \prod_{k\neq j}  \frac{(1-e^{2\pi i \mu }e^{2\pi i (z_k -z_j) })}{(1-e^{2\pi i (z_k -z_j) })}, \\
& \tilde{W}^{\langle 2 \rangle} _{j}(l;z) = e^{4\pi i z_j} \prod_{n=1}^4 (1- e^{2\pi i l_n} e^{\pi a_{-}} e^{-2\pi i z_j} )  \prod_{k\neq j}  \frac{(1- e^{-2\pi i \mu } e^{2\pi i (z_k -z_j) } )}{(1-e^{2\pi i (z_k -z_j) })}, \nonumber 
\end{align}
\begin{align}
& \tilde{U}^{\langle 2 \rangle} (h,l;z) = \sum _{n=1}^4 ( e^{2\pi i h_n }+ e^{2\pi i l_n }) \sum _{j=1}^N  e^{2\pi i z_j} -(e^{\pi a_- }+e^{-\pi a_- })\sum _{j=1}^N e^{4\pi i z_j} \\
& +e^{-2 \pi i \mu  } e^{-\pi a_- } (e^{2\pi i \mu  } -1)(e^{2\pi i \mu }e^{2 \pi a_- } -1) \sum _{1\leq j<k \leq N} e^{2\pi i z_j}e^{2\pi i z_k} \nonumber \\
& + \prod_{n=1}^4 e^{\pi i ( h_n + l_n )} \Big[ \sum _{n=1}^4 ( e^{-2\pi i h_n } + e^{-2\pi i l_n } ) \sum _{j=1}^N  e^{-2\pi i z_j}  - (e^{\pi a_- }+e^{-\pi a_- } )\sum _{j=1}^N  e^{-4\pi i z_j} \nonumber \\
& +e^{-2 \pi i \mu  } e^{-\pi a_- } (e^{2\pi i \mu  } -1)(e^{2\pi i \mu }e^{2 \pi a_- } -1)  \sum _{1\leq j <k \leq N}  e^{-2\pi i z_j} e^{-2\pi i z_k} \Big] . \nonumber
\end{align}
Note that it seems that this operator was also essentially obtained in \cite{vD}.

We apply the third degeneration.
\begin{prop} \label{prop:thirddegN}
In Eq.(\ref{eq:NE7}), we replace $z$ by $z-iR$, $h_n $ $(n=1,2)$ by $h_n -iR$, $h_n $ $(n=3,4)$ by $h_n +iR$, $l_n $ $(n=1,2,3,4)$ by $l_n -iR$  and take the limit $R\to +\infty $.
Then we have the operator
\begin{equation}
A^{\langle 3 \rangle} (h,l;z) = \sum _{j=1}^N ( V^{\langle 3 \rangle}_{j}(h;z)\exp(-ia_{-}\partial_{z_j})+W^{\langle 3 \rangle}_{j}(l;z)\exp(ia_{-}\partial_{z_j}) )+U^{\langle 3 \rangle}(h,l;z),
\end{equation}
where
\begin{align}
& V^{\langle 3 \rangle}_{j}(h;z) = e^{4\pi i  z_j} \prod_{n=1}^2 (1-e^{2\pi i h_n } e^{-\pi a_{-}}e^{-2\pi i  z_j}) \prod_{k\neq j}  \frac{(1-e^{2\pi i \mu }e^{2\pi i (z_k -z_j) })}{(1-e^{2\pi i (z_k -z_j)})}, \\
& W^{\langle 3 \rangle}_{j}(l;z) =e^{4\pi i z_j} \prod_{n=1}^4 (1- e^{2\pi i l_n} e^{\pi a_{-}} e^{-2\pi i z_j} )  \prod_{k\neq j}  \frac{(1- e^{-2\pi i \mu } e^{2\pi i (z_k -z_j) } )}{(1-e^{2\pi i ( z_k-z_j) })}, \nonumber 
\end{align}
\begin{align}
& U^{\langle 3 \rangle} (h,l;z) =\Big( \sum _{n=1}^2 e^{2\pi i h_n }+\sum _{n=1}^4 e^{2\pi i l_n }\Big) \sum _{j=1}^N  e^{2\pi i z_j} -(e^{\pi a_- }+e^{-\pi a_- })\sum _{j=1}^N e^{4\pi i z_j} \\
& \qquad +e^{-2 \pi i \mu  } e^{-\pi a_- } (e^{2\pi i \mu  } -1)(e^{2\pi i \mu }e^{2 \pi a_- } -1) \sum _{1\leq j<k \leq N} e^{2\pi i z_j}e^{2\pi i z_k} \nonumber \\
& \qquad +  e^{\pi i h_1}  e^{\pi i h_2}  ( e^{\pi i (h_3 - h_4 )} + e^{\pi i (h_4 - h_3 )})  \prod_{n=1}^4 e^{\pi i l_n} \sum _{j=1}^N  e^{-2\pi i z_j} .\nonumber
\end{align}
\end{prop}
By applying a gauge transformation, we have 
\begin{equation}
\tilde{A}^{\langle 3 \rangle} (h,l;z) = \sum _{j=1}^N ( \tilde{V}^{\langle 3 \rangle}_{j}(h;z)\exp(-ia_{-}\partial_{z_j})+ \tilde{W}^{\langle 3 \rangle}_{j}(l;z)\exp(ia_{-}\partial_{z_j}) )+ U^{\langle 3 \rangle} (h,l;z),
\label{eq:NE6}
\end{equation}
where
\begin{align}
& \tilde{V}^{\langle 3 \rangle} _{j}(h;z) = e^{-2\pi a_{-}} \prod_{n=1}^2 (1-e^{2\pi i h_n } e^{-\pi a_{-}}e^{-2\pi i  z_j}) \prod_{k\neq j}  \frac{(1-e^{2\pi i \mu }e^{2\pi i (z_k -z_j) })}{(1-e^{2\pi i (z_k -z_j) })}, \\
& \tilde{W}^{\langle 3 \rangle}_{j}(l;z) = e^{-2\pi a_{-}} e^{8\pi i z_j}\prod_{n=1}^4 (1- e^{2\pi i l_n} e^{\pi a_{-}} e^{-2\pi i z_j} )  \prod_{k\neq j}  \frac{(1- e^{-2\pi i \mu } e^{2\pi i (z_k -z_j)} )}{(1-e^{2\pi i (z_k -z_j) })}. \nonumber 
\end{align}

We apply the fourth degeneration.
\begin{prop} \label{prop:fourthdegN}
In Eq.(\ref{eq:NE6}), we replace $z$ by $z+iR$, $h_n $ $(n=1,2,3,4)$ by $h_n +iR$, $l_n $ $(n=1,2)$ by $l_n +iR$, $l_n $ $(n=3,4)$ by $l_n -iR$  and take the limit $R\to +\infty $.
Then we have the operator
\begin{equation}
A^{\langle 4 \rangle} (h,l;z) \equiv \sum _{j=1}^N ( V^{\langle 4 \rangle} _{j}(h;z)\exp(-ia_{-}\partial_{z_j})+W^{\langle 4 \rangle}_{j}(l;z)\exp(ia_{-}\partial_{z_j}) )+U^{\langle 4 \rangle}(h,l;z),
\end{equation}
where
\begin{align}
& V^{\langle 4 \rangle}_{j}(h;z) = e^{-2\pi a_{-}} \prod_{n=1}^2 (1-e^{2\pi i h_n } e^{-\pi a_{-}}e^{-2\pi i  z_j}) \prod_{k\neq j}  \frac{(1-e^{2\pi i \mu }e^{2\pi i (z_k -z_j) })}{(1-e^{2\pi i (z_k -z_j) })}, \\
& W^{\langle 4 \rangle}_{j}(l;z) =  e^{4\pi i z_j} e^{2\pi i (l_3 +l_4 )} \prod_{n=1}^2 (1- e^{2\pi i l_n} e^{\pi a_{-}} e^{-2\pi i z_j} )  \prod_{k\neq j}  \frac{(1- e^{-2\pi i \mu } e^{2\pi i (z_k -z_j)} )}{(1-e^{2\pi i (z_k -z_j) })}, \nonumber \\
& U^{\langle 4 \rangle}(h,l;z) = \nonumber \\
& \quad  (e^{2\pi i l_3 } + e^{2\pi i l_4 } ) \sum _{j=1}^N  e^{2\pi i z_j} +  e^{\pi i (h_1+ h_2)}  ( e^{\pi i (h_3 - h_4 )} + e^{\pi i (h_4 - h_3 )})  \prod_{n=1}^4 e^{\pi i l_n} \sum _{j=1}^N  e^{-2\pi i z_j} . \nonumber
\end{align}
\end{prop}
By a multiplication and a gauge transformation we have 
\begin{align}
\tilde{A}^{\langle 4 \rangle} (h,l;z) = \sum _{j=1}^N ( \tilde{V}^{\langle 4 \rangle}_{j}(h;z)\exp(-ia_{-}\partial_{z_j})+\tilde{W}^{\langle 4 \rangle}_{j}(l;z)\exp(ia_{-}\partial_{z_j}) )-U^{\langle 4 \rangle}(h,l;z),
\label{eq:ND5}
\end{align}
where
\begin{align}
& \tilde{V}^{\langle 4 \rangle}_{j}(h;z) = e^{2\pi i z_j} \prod_{n=1}^2 (1-e^{2\pi i h_n } e^{-\pi a_{-}}e^{-2\pi i  z_j}) \prod_{k\neq j} \frac{(1-e^{2\pi i \mu }e^{2\pi i (z_k -z_j) })}{(1-e^{2\pi i (z_k -z_j) })}, \\
& \tilde{W}^{\langle 4 \rangle}_{j}(l;z) = e^{2\pi i (l_3 +l_4 )} e^{2 \pi i z_j} \prod_{n=1}^2 (1- e^{2\pi i l_n} e^{\pi a_{-}} e^{-2\pi i z_j} ) \prod_{k\neq j}  \frac{(1- e^{-2\pi i \mu } e^{2\pi i (z_k -z_j)} )}{(1-e^{2\pi i (z_k -z_j)})}. \nonumber 
\end{align}

\section{Discussion} \label{sec:discuss}

We have found out that the degenerated Ruijsenaars-van Diejen operators of one variable appear by specializations of the linear $q$-difference equations related with the $q$-Painlev\'e equations of types $D_5^{(1)}$, $E_6^{(1)}$ and $E_7^{(1)}$.
Our results should be extended to the case of the $q$-Painlev\'e equation of type $E_8^{(1)}$ and the elliptic-difference Painlev\'e equation.
Note that Yamada and his collaborators found Lax pairs of the $q$-Painlev\'e equations of type $E_8^{(1)}$ and the elliptic-difference Painlev\'e equation \cite{Ye,Y,NTY}.
On Lax pairs of the elliptic-difference Painlev\'e equation, see also the papers by Rains and Ormerod \cite{Rai,OR}.

We propose other related problems.
Komori and Hikami proved existence of the commuting operators for the multivariable Ruijsenaars-van Diejen operator \cite{KH}.
The commuting operators of the multivariable degenerate operators should be clarified.

It is known that the Ruijsenaars-van Diejen operator of one variable admits $E_8$ symmetry \cite{Rui15}.
A kernel function plays important roles in \cite{Rui15}, because the Hilbert-Schmidt operator of the kernel function is used to built up Hilbert space features and is also used to establish the invariance of the discrete spectra under the $E_8$ Weyl group.  
The symmetry of the degenerate operators should also be studied well.
In particular, the kernel functions for the degenerate operators should be established.

\section*{Acknowledgements}
The author is grateful to Simon Ruijsenaars for valuable comments and fruitful discussions.
He would like to thank the university of Leeds, where most parts of this papar were accomplished.
He was supported by JSPS KAKENHI Grant Number JP26400122 and by Chuo University Overseas Research Program.


\begin{thebibliography}{9999}
\bibitem{Guz}
D.~Guzzetti, The elliptic representation of the general Painlev\'e VI equation. {\it Comm. Pure Appl. Math.} {\bf 55}  (2002), 1280--1363.
\bibitem{IKSY} K.~Iwasaki, H.~Kimura,S.~Shimomura and M.~Yoshida, From Gauss to Painlev\'e, Aspects of Mathematics, E16, Braunschweig: Friedr. Vieweg \& Sohn, 1991.
\bibitem{JS} M.~Jimbo, H.~Sakai, A $q$-Analog of the Sixth Painlev\'e Equation. {\it Lett. Math. Phys.}  {\bf 38} (1996), 145--154
\bibitem{KH} Y.~Komori and K.~Hikami, Quantum integrability of the generalized elliptic Ruijsenaars models, {\it J. Phys. A} {\bf 30} (1997), 4341--4364.
\bibitem{Man}
Yu.~I.~Manin, Sixth Painlev\'e equation, universal elliptic curve, and mirror of ${\mathcal P}^2$. {\it Geometry of differential equations,}  131--151, Amer. Math. Soc. Transl. Ser. 2, {\bf 186}, Adv. Math. Sci., 39, Amer. Math. Soc., Providence, RI, 1998.
\bibitem{NTY} M.~Noumi, S.~Tsujimoto, and Y.~Yamada, Pade interpolation for elliptic Painlev\'e equation, in "Symmetries, integrable systems and representations", Springer London, 463-482, 2013.
\bibitem{OR}
C.~M.~Ormerod and E.~M.~Rains, Commutation relations and discrete Garnier systems, {\it SIGMA} {\bf 12} (2016), paper 110, 50 pp.
\bibitem{Rai}
E.~M.~Rains, An isomonodromy interpretation of the hypergeometric solution of the elliptic Painlev\'e equation (and generalizations), {\it SIGMA} {\bf 7} (2011), papar 088, 24 pp.
\bibitem{RuiN} S.~N.~M.~Ruijsenaars, Integrable $BC_N$ analytic difference operators: Hidden parameter symmetries and eigenfunctions, in Proceedings Cadiz 2002 NATO Advanced Research Workshop "New trends in integrability and partial solvability", NATO Science Series {\bf 132}, 217--261, Kluwer, Dordrecht, 2004.
\bibitem{Rui09} S.~N.~M.~Ruijsenaars, Hilbert-Schmidt operators vs.integrable systems of elliptic Calogero-Moser type. I. The eigenfunction identities, {\it Commun.~Math.~Phys.} {\bf 286} (2009), 629--657.
\bibitem{Rui15} S.~N.~M.~Ruijsenaars, Hilbert-Schmidt operators vs.integrable systems of elliptic Calogero-Moser type IV. The relativistic Heun (van Diejen) case, {\it SIGMA} {\bf 11} (2015), paper 004.
\bibitem{Sak} H.~Sakai, Rational surfaces with affine root systems and geometry of the Painlev\'e equations, {\it Commun. Math. Phys.} {\bf 220} (2001), 165--221
\bibitem{SakP} H.~Sakai, Problem: Discrete Painlev\'e equations and their Lax form, {\it RIMS Kokyuroku Bessatsu} {\bf B2} (2007), 195--208
\bibitem{SL}
S.~Slavyanov and W.~Lay, Special Functions, Oxford Science Publications, Oxford University Press, Oxford, 2000.
\bibitem{TakHP} K.~Takemura, Heun equation and Painlev\'e equation, in Proceedings of Workshop on Elliptic Integrable Systems 2004 Kyoto, Rokko Lectures in Mathematics {\bf 18}, 305--322, 2005.
\bibitem{TakMH}
K.~Takemura, Middle convolution and Heun's equation, {\it SIGMA} {\bf 5} (2009), paper 040.
\bibitem{TakS}
K.~Takemura, Heun's differential equation (translation of Heun's differential equation (Japanese), {\it Sugaku} {\bf 60} (2008), 272--294),  Selected papers on analysis and differential equations, 45--68, Amer. Math. Soc. Transl. Ser. 2, {\bf 230}, Amer. Math. Soc., Providence, 2010.
\bibitem{vD} J.~F.~van Diejen, Difference Calogero-Moser systems and finite Toda chains, {\it J. Math. Phys.} {\bf 36} (1995), 1299--1323.
\bibitem{Ye} Y.~Yamada, A Lax formalism for the elliptic difference Painlev\'e equation,  {\it SIGMA} {\bf 5} (2009), paper 042.
\bibitem{Y} Y.~Yamada, Lax formalism for q-Painleve equations with affine Weyl group symmetry of type $E^{(1)}_n$, {\it Int. Math. Res. Notices} (2011) 3823--3838.
\bibitem{ZZ} A.~Zabrodin and A.~Zotov, Quantum Painlev\'e-Calogero correspondence for Painlev\'e VI. {\it J. Math. Phys.} {\bf 53} (2012), 073508, 19 pp.
\end{thebibliography}
\end{document}